\documentclass[conference]{IEEEtran}
\IEEEoverridecommandlockouts
\usepackage{cite}
\usepackage{amsmath,amssymb,amsfonts}
\usepackage{BOONDOX-cal}
\usepackage{mathrsfs}
\usepackage{bm}
\usepackage{soul}
\usepackage{algorithm}
\usepackage{algorithmic}
\usepackage{multirow}
\usepackage{verbatimbox}
\usepackage{booktabs}
\usepackage{graphicx}
\usepackage{subfigure}
\usepackage{siunitx}
\usepackage{empheq}
\usepackage{amsthm}
\usepackage{graphicx}
\usepackage{textcomp}
\usepackage{xcolor}
\usepackage[hidelinks]{hyperref}
\def\BibTeX{{\rm B\kern-.05em{\sc i\kern-.025em b}\kern-.08em
    T\kern-.1667em\lower.7ex\hbox{E}\kern-.125emX}}
\newtheorem{theorem}{Theorem}
\newtheorem{corollary}[theorem]{Corollary}
\newtheorem{lemma}[theorem]{Lemma}
\newtheorem{definition}{Definition}
\newtheorem{example}{Example}

\newtheorem{proposition}[theorem]{Proposition}
\newtheorem{assumption}{Assumption}

\begin{document}

\title{Robust Causal Learning for the Estimation of Average Treatment Effects
}

\author{\IEEEauthorblockN{Yiyan Huang\IEEEauthorrefmark{1}}
	\IEEEauthorblockA{\textit{School of Data Science}\\
		\textit{The City University of Hong Kong}\\
		yiyhuang3-c@my.cityu.edu.hk}\\
			\IEEEauthorblockN{Xing Yan}
	\IEEEauthorblockA{\textit{Institute of Statistics and Big Data}\\
		\textit{Renmin University of China}\\
		xingyan@ruc.edu.cn}\\
	\IEEEauthorblockN{Dongdong Wang}
	\IEEEauthorblockA{\textit{JD Digits}\\
		wangdongdong9@jd.com}
	\and
	\IEEEauthorblockN{Cheuk Hang Leung\IEEEauthorrefmark{1}}
	\IEEEauthorblockA{\textit{School of Data Science}\\
		\textit{The City University of Hong Kong}\\
		chleung87@cityu.edu.hk}\\
				\IEEEauthorblockN{Shumin Ma \IEEEauthorrefmark{3}}
	\IEEEauthorblockA{\textit{Division of Science and Technology}\\
		\textit{BNU-HKBU United International College}\\
		shuminma@uic.edu.cn}\\
		\IEEEauthorblockN{Zhixiang Huang}
	\IEEEauthorblockA{\textit{JD Digits}\\
		huangzhixiang@jd.com}
		\and
	\IEEEauthorblockN{Qi Wu \IEEEauthorrefmark{2}}
	\IEEEauthorblockA{\textit{School of Data Science}\\
		\textit{The City University of Hong Kong}\\
		qiwu55@cityu.edu.hk}\\
	\IEEEauthorblockN{Zhiri Yuan}
\IEEEauthorblockA{\textit{JD Digits-CityU Joint Lab}\\
	\textit{The City University of Hong Kong}\\
	yuanzhiri2012@gmail.com}
	\thanks{\IEEEauthorrefmark{1} Co-first authors are in alphabetical order.
		}
	\thanks{\IEEEauthorrefmark{2} Qi Wu is the corresponding author.
		}
	\thanks{\IEEEauthorrefmark{3} Shumin Ma is also with Guangdong Provincial Key Laboratory of Interdisciplinary Research and Application for Data Science, BNU-HKBU United International College.
}
}

\maketitle
\begin{abstract}
Many practical decision-making problems in economics and healthcare seek to estimate the average treatment effect (ATE) from observational data. The Double/Debiased Machine Learning (DML) is one of the prevalent methods to estimate ATE in the observational study. However, the DML estimators can suffer an error-compounding issue and even give an extreme estimate when the propensity scores are misspecified or very close to 0 or 1. Previous studies have overcome this issue through some empirical tricks such as propensity score trimming, yet none of the existing literature solves this problem from a theoretical standpoint. In this paper, we propose a Robust Causal Learning (RCL) method to offset the deficiencies of the DML estimators. Theoretically, the RCL estimators i) are as consistent and doubly robust as the DML estimators, and ii) can get rid of the error-compounding issue. Empirically, the comprehensive experiments show that i) the RCL estimators give more stable estimations of the causal parameters than the DML estimators, and ii) the RCL estimators outperform the traditional estimators and their variants when applying different machine learning models on both simulation and benchmark datasets. 
\end{abstract}


\section{Introduction}\label{sec:introduction}
Causal inference is ubiquitous for decision-making problems in various areas such as Healthcare \cite{glass2013causal, hill2013assessing, alaa2017bayesian} and Economics \cite{belloni2014inference, farrell2015robust, chernozhukov2018double}. At the core of causal machine learning, estimating the average treatment effect (ATE) from observational data is challenging because some features (covariates) can influence both treatment and outcome in most practical circumstances. For example, factors such as regions and races (covariates) can affect both the vaccination (treatment) assignment and the post-vaccination infection rate (outcome). To obtain a clean ATE, one can conduct the Randomized Controlled Trials (RCTs). RCTs are regarded as the gold standard to evaluate ATE, whereas conducting RCTs is often expensive and time-consuming. As a result, more and more researchers tend to estimate ATE from observational data.

In the observational study, classical causal learning methods concerning ATE estimations mainly include regression adjustment methods and re-weighting methods (see more details in \cite{10.1145/3444944}). Regression adjustment methods require an estimated feature-outcome relation (aka the outcome model) and directly average the predicted potential outcomes over the whole population to estimate ATE, so the associated estimator is called the direct regression (DR) estimator. The conundrum of the DR estimator is that it overlooks the probabilistic impact of the covariates on the treatment assignment (i.e., the propensity score) and hence often results in biased estimations of ATE unless the outcome model is estimated accurately. Re-weighting methods mimic the principle of RCTs to make the re-weighted instances look like they receive alternative treatment. The Inverse Probability Weighting (IPW) is one of the prevalent re-weighting strategies. It involves the propensity scores rather than the outcome model. Nevertheless, the IPW estimator is sensitive to the estimation of propensity scores and even leads to high variance estimates. Such occasion often occurs when the estimated propensity scores are close to 0 or 1. This is called an \textit{error-compounding issue}.

The Debiased Machine Learning (DML) method, which is exploited by \cite{chernozhukov2018double} based on \cite{neyman1979c}, offsets the shortcomings of classical causal learning methods. The DML method combines the two classical approaches to ensure that the corresponding estimator is accurate as long as either the outcome model or the propensity score, but not necessarily both, is correctly specified (see \cite{robins2017minimax}, \cite{robins2008higher}, \cite{mukherjee2017semiparametric}, \cite{kang2007demystifying}, \cite{van2014higher} and the references therein). This notable merit is well known as the ``doubly robust'' property. However, the DML estimator can still suffer the error-compounding issue since the inverse term of the propensity score is still present. Inevitably, propensity scores usually fail to be correctly specified in practice, especially when the distribution of the treated group is substantially different from that of the controlled group (see, for example, \cite{li2018balancing, busso2014new, knaus2020double, austin2015moving, dudik2011doubly}). This observation motivates us to go beyond the DML estimator and construct estimators that are more robust to the misspecification of the estimated propensity scores.


In this paper, we propose a Robust Causal Learning (RCL) method to establish the RCL estimators of the ATE. The contributions are summarized as follows: 
\begin{enumerate}
	\item [1.] Our RCL estimators robustly ease the error-compounding issue exhibited by the DML estimator since the propensity scores in the RCL estimators are no longer in an inverse form.
	\item [2.] The RCL estimators inherit the consistency and doubly robust property of the DML estimator.
	\item [3.] The RCL methodology to construct an ATE estimator can also be applied to establish other prevalent treatment effect estimators.
	\item [4.] Extensive experiments show that the proposed RCL method achieves superior performance than DR, IPW, DML, and their variants across various combinations of machine learning regressors and classifiers.
\end{enumerate}


The rest of the paper is organized as follows. Section \ref{sec:background} introduces the problem setup and the background of orthogonal scores. Section \ref{sec:$k^{th}$-order condition and assumptions} presents the main theoretical results, including the RCL score in Theorem \ref{thm:$k^{th}$-order orthogonal condition score function} and the RCL estimator in Corollary \ref{corollary:$k^{th}$-order orthogonal condition estimator}. Section \ref{sec:experiment} reports extensive experimental results on simulation datasets and benchmark datasets. Due to the space limit, we defer all proof details and the code for reproducing our experiments to the full paper version.

\section{Preliminaries}\label{sec:background}
\subsection{The Problem Setup}
In this paper, we consider the potential outcome framework \cite{rubin1974estimating, rubin2005causal} to study ATE. Let {\small $\mathbf{Z}$} be the covariates (aka confounders), and {\small $D$} be the treatment variable which can take values from {\small $\{d^{1},\dots, d^{n}\}$}. We denote {\small $Y$} as the outcome variable (aka the response), and {\small $Y^{i}$} represents the \textit{potential outcome} under the treatment {\small $d^{i}$}. If the observed treatment is {\small $d^{i}$}, then the \textit{factual outcome} {\small $Y^{F}$} equals {\small $Y^{i}$}. We denote {\small $\{w_m=(y_m, d_m, \mathbf{z}_m)\}_{m=1}^{N}$} as observed $N$ realizations of the i.i.d. random variables {\small $\{W_m=(Y_m, D_m, \mathbf{Z}_m)\}_{m=1}^{N}$}. Given the true causal parameter {\small $\theta^{i}:=\mathbb{E}[Y^{i}]$}, the target quantity ATE between treatment $d^{i}$ and treatment $d^{j}$ is defined as 
\begin{equation}
{\small
\begin{aligned}
\theta^{i,j}=\theta^{i}-\theta^{j}.
\end{aligned}
}
\end{equation}
Identifying the ATE {\small $\theta^{i,j}$} under the potential outcome framework requires some fundamental assumptions to ensure that {\small $\theta^{i}$} and {\small $\theta^{j}$} are identifiable. Thus, we impose the following assumptions as stated in existing causal inference works.

\begin{assumption}[Stable Unit Treatment Value Assumption (SUTVA)]
	The potential outcomes for any individual do not vary regardless of the treatment status of other individuals.
\end{assumption}
\begin{assumption}[Ignorability]
	Given the covariates {\small $\mathbf{Z}$}, the potential outcome {\small $Y^{i}$} is independent to the treatment assignment {\small $D$}, i.e., {\small $(Y^{1},\cdots, Y^{n})\perp\!\!\!\perp D\mid \mathbf{Z} \;\; \forall i$}.
\end{assumption}

\begin{assumption}[Positivity]
	Treatment assignment is not deterministic regardless of the values covariates {\small $\mathbf{Z}$} take, i.e., {\small $0<\mathbb{P}\{D=d^{i}\mid \mathbf{Z}=\mathbf{z}\}<1$} {\small $\forall\;i$} and {\small $\forall\;\mathbf{z}$}.
\end{assumption}

\begin{assumption}[Consistency]
If an individual receives treatment {\small $d^{i}$}, his factual outcome {\small $Y^{F}$} is equal to the potential outcome under treatment {\small $d^{i}$}, i.e., {\small $Y^{F}=Y^{i}$} if {\small $D=d^{i}$}.
\end{assumption}

These assumptions guarantee that the ATE can be inferred if we specify the relation {\small $\mathbb{E}\left[Y \mid D, \mathbf{Z}\right]$}, which is equivalent to estimating {\small $g^{i}(\mathbf{Z})$} for each {\small $i \in \{1, \dots, n\}$} in the generalized propensity score model setting\footnote{This model setting allows {\small $D$} to be multi-valued. It can be reduced to the ``iteractive model" in \cite{chernozhukov2018double} once the treatment {\small $D$} takes binary values.} \cite{imbens2000role,tu2013using} \eqref{eqt:model 1}:
{\small	
	\begin{equation}\label{eqt:model 1}
			\begin{aligned}		
			Y^{i}&=g^{i}(\mathbf{Z})+\xi^{i},\; &&\mathbb{E}\left[\xi^{i}\mid D, \mathbf{Z}\right]=0 && a.s.,\\
			\mathbf{1}_{\{D=d^{i}\}}&=\pi^{i}(\mathbf{Z})+\nu^{i},\; &&\mathbb{E}\left[\nu^{i}\mid\mathbf{Z}\right]=0 && a.s..
		\end{aligned}	
	\end{equation}
}\noindent
Here, {\small $g^i(\cdot)$} and {\small $\pi^i(\cdot)$} are true \textit{nuisance parameters}. {\small $\xi^{i}$} and {\small $\nu^{i}$} are the noise terms. {\small $\pi^i(\mathbf{Z})=\mathbb{E}\left[\mathbf{1}_{\{D=d^{i}\}} \mid \mathbf{Z}\right]$} is known as the generalized propensity score (GPS) with multi-valued treatment variables. Finally, the true causal parameter {\small $\theta^i$} for {\small $i \in \{1, \dots, n\}$} can be computed by {\small $\theta^i:=\mathbb{E}\left[Y^{i}\right]=\mathbb{E}\left[g^i(\mathbf{Z})\right]$}  and the true ATE can be computed by {\small $\theta^{i,j}=\theta^{i}-\theta^{j}$}.

\subsection{Non-Orthogonal Scores and Orthogonal Scores}
We aim to estimate the true causal parameters $\theta^{i}$ given $N$ i.i.d. samples $\{W_m=(\mathbf{Z}_m, D_m, Y_m)\}^{N}_{m=1}$. According to \cite{chernozhukov2018double}, the standard procedure to obtain the estimated causal parameter $\hat{\theta}^i$ is: 1) getting the estimated nuisance parameters $\hat{\rho}$, e.g., $\hat{\rho}=(\hat{g}^i, \hat{\pi}^i)$; 2) constructing a \textit{score} that satisfies the \textit{moment condition} (Definition \ref{def:moment condition}); 3) establishing the estimator of $\theta^i$, which is solved from the moment condition \eqref{eqt:moment condition}.

\begin{definition}[\textit{\textbf{Moment Condition}}]\label{def:moment condition}
	Let {\small $W=(\mathbf{Z}, D, Y)$} and {\small $\theta$} be the true causal parameter with {\small $\vartheta$} being a causal parameter that lies in the causal parameter set. Denoting the nuisance parameters as {\small $\varrho$} and the true nuisance parameters as {\small $\rho$}, we say a score {\small $\psi(W,\vartheta,\varrho)$} satisfies the \textit{moment condition} if
	\begin{equation}
		\begin{aligned}\label{eqt:moment condition}
			\mathbb{E}\left[\psi(W,\vartheta,\varrho)|_{\vartheta=\theta,\;\varrho=\rho}\right]=0.
		\end{aligned}
	\end{equation}
\end{definition}
The moment condition guarantees that the estimator derived from the score is unbiased if the nuisance parameters equal the true ones. Here, we give the scores which satisfy the moment condition of two classical causal learning methods (DR and IPW) introduced before.
\begin{example}[\textit{The Score and Estimator for DR}]\label{example:DR}
Let {\small $\varrho=\mathcal{g}^{i}$} and $\rho=g^{i}$. In the DR method, the score {\small $\psi_{DR}^{i}(W,\vartheta,\varrho)$} satisfying the moment condition and the associated estimator {\small $\hat{\theta}_{DR}^{i}$} are
{\small 	\begin{equation*}
		\begin{aligned}
			\psi_{DR}^{i}(W,\vartheta,\varrho)=\vartheta-\mathcal{g}^{i}(\mathbf{Z}); \quad 
			\hat{\theta}_{DR}^{i}=\frac{1}{N}\underset{m=1}{\overset{N}{\sum}}\hat{g}^{i}(\mathbf{Z}_{m}).
		\end{aligned}
	\end{equation*}}
\end{example}
\begin{example}[\textit{The Score and Estimator for IPW}]\label{example:IPW}	
Let {\small $\varrho=a_{i}$} and {\small $\rho=\pi^{i}$}. In the IPW method, the score {\small $\psi_{IPW}^{i}(W,\vartheta,\varrho)$} satisfying the moment condition and the associated estimator {\small $\hat{\theta}_{IPW}^{i}$} are
		{\small 	\begin{equation*}
				\begin{aligned}
					\psi_{IPW}^{i}(W,\vartheta,\varrho)=\vartheta-\frac{Y\mathbf{1}_{\{D=d^{i}\}}}{a_{i}(\mathbf{Z})}; \
					\hat{\theta}_{IPW}^{i}=\frac{1}{N}\underset{m=1}{\overset{N}{\sum}}\frac{Y_{m}\mathbf{1}_{\{D_{m}=d^{i}\}}}{\hat{\pi}^{i}(\mathbf{Z}_{m})}.
				\end{aligned}
		\end{equation*}}
\end{example}
Generally, the estimators established from the scores in Example 1 might be invalid unless {\small $\hat{g}^{i}$} and {\small $\hat{\pi}^{i}$} estimate {\small $g^{i}$} and {\small $\pi^{i}$} well. To obtain robust estimators, \cite{chernozhukov2018double} suggest that we should construct scores which satisfy the Orthogonal Condition (Definition \ref{def:orthogonal condition}) apart from the moment condition.

\begin{definition}[\textit{\textbf{Orthogonal Condition}}]\label{def:orthogonal condition}
	Suppose that the nuisance parameters and the true nuisance parameters are {\small $\gamma$}-dimensional tuples, i.e., {\small $\varrho=(\mathcal{h}_{1},\cdots, \mathcal{h}_{\gamma})$} and {\small $\rho=(h_{1},\cdots, h_{\gamma})$}. Given {\small $S\subseteq \mathbb{Z}_{\geq 0}^{\gamma}$}, we say a score {\small $\psi(W,\vartheta,\varrho)$} satisfies the \textit{orthogonal condition} if
\begin{equation}
{\small
\begin{aligned}
\mathbb{E}\left[\mathrm{\mathbf{D}}^{\alpha}\psi(W,\vartheta,\varrho)\mid_{\vartheta=\theta,\;\varrho=\rho}\mid\mathbf{Z}\right]=0\; \forall \alpha\in S. \label{eqt:$k^{th}$-order orthogonal condition}
\end{aligned}
}
\end{equation}
{\small $S$} can be any subset of {\small $\mathbb{Z}_{\geq 0}^{\gamma}$}. Throughout the paper, for some positive integer {\small $k$}, we define {\small $S$} as
\begin{equation}
{\small
\begin{aligned}\label{eqt:S definition}
S=\{\alpha\in\mathbb{Z}_{\geq 0}^{\gamma}\mid \left\|\alpha\right\|_{1}\leq k\}
\end{aligned}
}
\end{equation}
and {\small $\mathrm{\mathbf{D}}^{\alpha}\psi(W,\vartheta,\varrho)=\partial^{\alpha_{1}}_{\mathcal{h}_{1}}\partial^{\alpha_{2}}_{\mathcal{h}_{2}}\cdots\partial^{\alpha_{\gamma}}_{\mathcal{h}_{\gamma}}\psi(W,\vartheta,\mathcal{h}_{1},\cdots, \mathcal{h}_{\gamma})$}.
\end{definition}
The orthogonal condition ensures that the established estimators can still be valid even though some nuisance parameters are misspecified (see \cite{chernozhukov2018double, robins2008higher, mackey2018orthogonal} for more details). Below we demonstrate how to utilize {\small $\mathbf{D}^{\alpha}\psi(W,\vartheta,\varrho)$} to justify that the scores in Examples \ref{example:DR} and  \ref{example:IPW} violate the orthogonal condition. Suppose $k$ in \eqref{eqt:S definition} is $1$.
{\small \begin{equation*}
	\begin{aligned}
	&\mathbf{D}^{(1)}\psi_{DR}^{i}(W,\vartheta,\varrho)=\partial_{\mathcal{g}^{i}}\psi_{DR}^{i}(W,\vartheta,\varrho)=-1;\\
	&\mathbf{D}^{(1)}\psi_{IPW}^{i}(W,\vartheta,\varrho)=\partial_{a_{i}}\psi_{IPW}^{i}(W,\vartheta,\varrho)=\frac{Y\mathbf{1}_{\{D=d^{i}\}}}{a_{i}(\mathbf{Z})^{2}};\\
	&\mathbb{E}\left[\mathbf{D}^{(1)}\psi_{DR}^{i}(W,\vartheta,\varrho)\mid_{\vartheta=\theta^i,\;\varrho=g^i} \mid \mathbf{Z}\right]=-1 \neq 0;\\
	&\mathbb{E}\left[\mathbf{D}^{(1)}\psi_{IPW}^{i}(W,\vartheta,\varrho)\mid_{\vartheta=\theta^i,\;\varrho=\pi^i} \mid \mathbf{Z}\right]=\mathbb{E}\left[\frac{Y\mathbf{1}_{\{D=d^{i}\}}}{\pi^{i}(\mathbf{Z})^{2}} \mid \mathbf{Z}\right] \neq 0.\\
	\end{aligned}
\end{equation*}}\noindent
The above calculations show that {\small $\psi_{DR}^{i}(W,\vartheta,\varrho)$} and {\small $\psi_{IPW}^{i}(W,\vartheta,\varrho)$} do not satisfy the orthogonal condition. The scores are usually termed as the \textit{non-orthogonal scores}. As a consequence, their associated estimators are not ``doubly robust". To obtain a doubly robust estimator, \cite{chernozhukov2018double} propose the DML method to construct the DML score.
\begin{example}[\textit{The Score and Estimator for DML}]\label{example:DML}	
Let {\small $\varrho=(\mathcal{g}^{i}, a_{i})$} and {\small $\rho=(g^{i},\pi^{i})$}. In the DML method, the score {\small $\psi_{DML}^{i}(W,\vartheta,\varrho)$} that satisfies both the moment condition and orthogonal condition and the associated estimator {\small $\hat{\theta}_{DML}^{i}$} are
	{\small
		\begin{equation*}
						\begin{aligned}
				&\psi_{DML}^{i}(W,\vartheta,\varrho)=\vartheta-\mathcal{g}^{i}(\mathbf{Z})-\frac{\bm{1}_{\{D=d^{i}\}}}{a_{i}(\mathbf{Z})}(Y-\mathcal{g}^{i}(\mathbf{Z}));\\
				&\hat{\theta}_{DML}^{i}=\frac{1}{N}\overset{N}{\underset{m=1}{\sum}}\hat{g}^{i}(\mathbf{Z}_{m})+\frac{1}{N}\overset{N}{\underset{m=1}{\sum}}\frac{\bm{1}_{\{D_{m}=d^{i}\}}(Y_{m}-\hat{g}^{i}(\mathbf{Z}_{m}))}{\hat{\pi}^{i}(\mathbf{Z}_{m})}.
			\end{aligned}
		\end{equation*}
}

\end{example}
We can prove that {\small $\psi_{DML}^{i}(W,\vartheta,\varrho)$} satisfies the orthogonal condition when {\small $k=1$} in \eqref{eqt:S definition} (see \cite{chernozhukov2018double} for detailed derivations) following similar calculation processes for DR and IPW. {\small $\psi_{DML}^{i}(W,\vartheta,\varrho)$} is therefore termed as the \textit{orthogonal score}. The orthogonal condition assures that the DML estimator is doubly robust, i.e., the estimator is locally unbiased and consistent as long as either {\small $g^{i}$} or {\small $\pi^{i}$} is correctly specified. Despite the doubly robust property, the DML estimator still suffers an error-compounding issue once the encompassed inverse propensity score is slightly misspecified for some data points. In real applications, one seldom encounters a situation that propensity scores are correctly estimated for all individuals. This dilemma motivates us to construct scores such that 1) the scores are orthogonal scores, i.e., they satisfy the moment condition (Definition \ref{def:moment condition}) and the orthogonal condition (Definition \ref{def:orthogonal condition}); 2) the estimators established from the scores can stabilize the estimation error due to the misspecifications on propensity scores.

In the upcoming section, we will introduce a novel method, the Robust Causal Learning (RCL) method, to overcome the difficulties encountered by DR, IPW, and DML methods.

\section{The Proposed Method}\label{sec:$k^{th}$-order condition and assumptions}
This section shows our main theoretical results. First, Section \ref{sec:Construction of the RCL Score} demonstrates the RCL scores. Then Section \ref{sec:Establishment of the RCL estimators} presents the detailed construction of the RCL estimators with an algorithm that describes how to obtain an estimate of {\small $\theta^{i}$} from observational data using the proposed RCL method.


\subsection{Construction of The RCL Score}\label{sec:Construction of the RCL Score}
In this paper, we construct an orthogonal score, the RCL score, to derive an estimator of {\small $\theta^{i}$} along the lines of orthogonal machine learning works (e.g.,  \cite{mackey2018orthogonal, chernozhukov2018double}). The relevant result is stated in Theorem \ref{thm:$k^{th}$-order orthogonal condition score function}.
\begin{theorem}[\textbf{\textit{RCL score}}]\label{thm:$k^{th}$-order orthogonal condition score function}
Suppose {\small $\varrho$} and {\small $\rho$} are {\small $2$}-dimensional tuples such that {\small $\varrho=(\mathcal{g}^{i},a_{i})$} and {\small $\rho=(g^{i},\pi^{i})$}. Let {\small $r$}, {\small $k$} be integers s.t. $1\leq k\leq r$. Assume the local moments {\small $\mathbb{E}\left[(\nu^{i})^{r}\mid\mathbf{Z}\right]\neq 0$} and {\small $\left|\mathbb{E}\left[(\nu^{i})^{q}\mid\mathbf{Z}\right]\right|<\infty\;a.s.\;\forall\;1\leq q\leq r$}. Under the assumptions on nuisance parameters and noise terms stated in \cite{chernozhukov2018double} and \cite{mackey2018orthogonal}, the RCL score {\small $\psi_{RCL}^{i}(W,\vartheta,\varrho)$} that satisfies the moment condition and the orthogonal condition is
{\small
\begin{align}
\psi_{RCL}^{i}(W,\vartheta,\varrho)&=\vartheta-\mathcal{g}^{i}(\mathbf{Z})-(Y^{i}-\mathcal{g}^{i}(\mathbf{Z}))A(D,\mathbf{Z};a_{i}). \label{eqt:$k^{th}$-order orthogonal condition score function}
\end{align}
}\noindent
Given an integer {\small $r$} and an integer {\small $k$}, we have
\begin{subequations}
\begin{equation}
{\small
\begin{aligned}
&A(D,\mathbf{Z};a_{i})=\bar{b}_{r}\left[\mathbf{1}_{\{D=d^{i}\}}-a_{i}(\mathbf{Z})\right]^{r}+\\
&\;\;\mathbf{1}_{\{k\neq 1\}}\left[\underset{q=1}{\overset{k-1}{\sum}}b_{q}\left(\left[\mathbf{1}_{\{D=d^{i}\}}-a_{i}(\mathbf{Z})\right]^{q}-\mathbb{E}\left[(\nu^{i})^{q}\mid\mathbf{Z}\right]\right)\right],
\label{eqt:$k^{th}$-order orthogonal condition score function 2}
\end{aligned}
}
\end{equation}
where {\small $\bar{b}_{r}=\frac{1}{\mathbb{E}\left[(\nu^{i})^{r}\mid\mathbf{Z}\right]}$} and the coefficient {\small $b_{q}$} is computed by descending order for {\small $q \in \{k-1,\dots,1\}$}:
\begin{equation}
{\small
\begin{aligned}
b_{q}&=-\bar{b}_{r}\binom{r}{q}\mathbb{E}\left[(\nu^{i})^{r-q}\mid\mathbf{Z}\right]\\
&\quad-\overset{k-1-q}{\underset{u=1}{\sum}}b_{q+u}\binom{q+u}{q}\mathbb{E}\left[(\nu^{i})^{u}\mid\mathbf{Z}\right].
\end{aligned}\label{eqt:kth_order orthogonal condition score function 2 coefficient}
}
\end{equation}
\end{subequations}
\end{theorem}
From \eqref{eqt:$k^{th}$-order orthogonal condition score function 2}, we can observe that $a_{i}(\cdot)$, the nuisance parameter of the propensity score, is no longer in an inverse form for the RCL score. As a consequence, the established RCL estimators from \eqref{eqt:$k^{th}$-order orthogonal condition score function} can avoid the error-compounding issue. Simultaneously, the RCL scores are orthogonal scores, so the RCL estimators are as doubly robust as the DML estimator.

\subsection{Establishment of the RCL estimators}\label{sec:Establishment of the RCL estimators}
In this part, we will go into detail about the establishment of the RCL estimators. To begin with, we can solve the estimator $\tilde{\theta}^{i}$ from \eqref{eqt:moment condition} using the emprical version of the moment condition for the RCL score \eqref{eqt:$k^{th}$-order orthogonal condition score function}:
\begin{subequations}
	{\small
		\begin{align}
			\tilde{\theta}^{i}=&\frac{1}{N}\underset{m=1}{\overset{N}{\sum}}g^{i}(\mathbf{Z}_{m})\label{tilde_theta_1} \\
			&+\frac{1}{N}\underset{m=1}{\overset{N}{\sum}}(Y^{i}_{m}-g^{i}(\mathbf{Z}_{m}))A(D_{m},\mathbf{Z}_{m};\pi^{i}). \label{tilde_theta_2}
		\end{align}
	}\noindent
\end{subequations}
Equation \eqref{tilde_theta_1} is referred to as the DR estimator when the true nuisance parameter {\small $g^i$} is replaced by the estimated one {\small $\hat{g}^i$}. Equation \eqref{tilde_theta_2} can then be divided into two parts:
\begin{subequations}
	{\small
		\begin{align}
\eqref{tilde_theta_2} &=\frac{1}{N}\underset{m\in\mathscr{I}}{\overset{}{\sum}}(Y^{i}_{m}-g^{i}(\mathbf{Z}_{m}))A(D_{m},\mathbf{Z}_{m};\pi^{i})\label{tilde_theta_2_1} \\
&+\frac{1}{N}\underset{m\in\mathscr{I}^{c}}{\overset{}{\sum}}(Y^{i}_{m}-g^{i}(\mathbf{Z}_{m}))A(D_{m},\mathbf{Z}_{m};\pi^{i})\label{tilde_theta_2_2},
		\end{align}
	}\noindent
\end{subequations}
where {\small $\mathscr{I}$} is the sample set in which the units are all treated with {\small $d^{i}$} while {\small $\mathscr{I}^c$} is the sample set in which the units are not treated with {\small $d^{i}$}. It is obvious that \eqref{tilde_theta_1} and \eqref{tilde_theta_2_1} can be directly calculated from observational data, whereas \eqref{tilde_theta_2_2} that contains the counterfactual outcomes is unavailable to compute in a direct manner. Instead of pursuing the unobservable counterfactuals, we realize that given {\small $i \in \{ 1,\dots,n \}$},
{\small \begin{equation*}
	\begin{aligned}
			&\mathbb{E}\left[(Y^i-g^i(\mathbf{Z}))A(D,\mathbf{Z};\pi^{i}) \mid D=d^j\right]\\
			&=\mathbb{E}\left[\mathbb{E}\left[\xi^{i} A(D,\mathbf{Z};\pi^{i}) \mid D=d^j, \mathbf{Z}\right]\mid D=d^j\right]\\
			&=\mathbb{E}\left[A(d^j,\mathbf{Z};\pi^{i})\mathbb{E}\left[\xi^{i} \mid D=d^j, \mathbf{Z}\right]\mid D=d^j\right]=0
	\end{aligned}
\end{equation*}}\noindent
holds for {\small $\forall j \in \{ 1,\dots,n \}$}. Thus, the sample mean of {\small $(Y^i-g^i(\mathbf{Z}))$} equals zero regardless the samples come from {\small $\mathscr{I}$} or {\small $\mathscr{I}^c$}. This observation allows us to replace the sample mean of the counterfactuals in \eqref{tilde_theta_2_2} with that of the factual ones. To be specific, we first define the set {\small $\mathcal{A}$} such that
\begin{equation}
{\small
\begin{aligned}
\mathcal{A}=\{Y^{i}_{m}-g^{i}(\mathbf{Z}_{m})\mid m\in\mathscr{I}\}. \label{eqt:random picking realization set random variable}
\end{aligned}
}
\end{equation}
Then, a replaced estimator of \eqref{tilde_theta_2_2} is obtained as follows:
\begin{enumerate}
	\item For the {\small $m^{\mathrm{th}}$} unit in the set {\small $\mathscr{I}^{c}$}, pick an element $\xi_{m}^{i}$ from {\small $\mathcal{A}$} and multiply it by {\small $A(D_{m},\mathbf{Z}_{m};\pi^{i})$}. Repeat the process until we go through all the individuals in the set {\small $\mathscr{I}^{c}$};
	\item Compute $\frac{1}{N}\underset{m\in\mathscr{I}^{c}}{\overset{}{\sum}}\xi_{m}^{i}A(D_{m},\mathbf{Z}_{m};\pi^{i})$;
	\item Repeat above steps {\small $R$} times to eliminate the randomness brought by the random picking procedure and return the substitute estimator {\small $\frac{1}{R}\overset{R}{\underset{u=1}{\sum}}\left[\frac{1}{N}\underset{m\in\mathscr{I}^{c}}{\overset{}{\sum}}\xi_{m,u}^{i} A(D_{m},\mathbf{Z}_{m};\pi^{i})\right]$}.
\end{enumerate}
Consequently, \eqref{tilde_theta_2_2} can be inferred indirectly from observational data. With \eqref{tilde_theta_1} and \eqref{tilde_theta_2_1}, the RCL estimator of {\small $\theta^{i}$} is finally established in Corollary \ref{corollary:$k^{th}$-order orthogonal condition estimator}.

\begin{corollary}[\textbf{\textit{RCL estimator}}]\label{corollary:$k^{th}$-order orthogonal condition estimator}
Let {\small $R\in\mathbb{Z}^{+}$}, {\small $(\hat{g}^{i},\hat{\pi}^{i})$} be the estimates of {\small $(g^{i},\pi^{i})$}, {\small $\hat{\mathcal{A}}$} be {\small $\mathcal{A}$} by replacing {\small $g^{i}$} with {\small $\hat{g}^{i}$} in \eqref{eqt:random picking realization set random variable}, and {\small $\hat{\xi}_{m,u}^{i}$} be the element that is randomly selected from the set {\small $\hat{\mathcal{A}}$} in the {\small $u^\mathrm{th}$} of $R$ repeated selections. The RCL estimator {\small $\hat{\theta}^{i}_{RCL}$} is given by
\begin{equation}
{\small
\begin{aligned}\label{eqt:$k^{th}$-order orthogonal condition final estimator 3-replaced}
\hat{\theta}^{i}_{RCL}=&\underbrace{\frac{1}{N}\underset{m=1}{\overset{N}{\sum}}\hat{g}^{i}(\mathbf{Z}_{m})}_{(a)}+\underbrace{\frac{1}{N}\underset{m\in\mathscr{I}}{\overset{}{\sum}}(Y^{i}_{m}-\hat{g}^{i}(\mathbf{Z}_{m}))A(D_{m},\mathbf{Z}_{m};\hat{\pi}^{i})}_{(b)}\\
&+\underbrace{\frac{1}{R}\overset{R}{\underset{u=1}{\sum}}\left[\frac{1}{N}\underset{m\in\mathscr{I}^{c}}{\overset{}{\sum}}\hat{\xi}_{m,u}^{i} A(D_{m},\mathbf{Z}_{m};\hat{\pi}^{i})\right]}_{(c)}.
\end{aligned}
}
\end{equation}
\end{corollary}
The proposed RCL estimator {\small $\hat{\theta}^{i}_{RCL}$} is a consistent estimator of {\small $\theta^{i}$} if {\small $(\hat{g}^{i},\hat{\pi}^{i})$} satisfy the assumptions stated in \cite{mackey2018orthogonal} and \cite{chernozhukov2018double}. Due to the space limit, the proofs of Theorem \ref{thm:$k^{th}$-order orthogonal condition score function} and the consistency of {\small $\hat{\theta}^{i}_{RCL}$} can be seen in the full paper version. We also outline the procedures of estimating {\small $\theta^{i}$} from observational data using the proposed RCL method in Algorithm \ref{alg:$k^{th}$-order orthogonal condition estimate}. Note that if the whole dataset is split into the training set and the test set, Step 2 will be only conducted on the training set, while Step 3 - Step 8 can be performed to obtain the estimates of {\small $\hat{\theta}^{i}_{RCL}$} on both the training set and the test set. The running complexity of our algorithm is at most {\small $O(NR)$}.
\begin{algorithm}[t]
\caption{Algorithm of obtaining an estimate of $\theta^{i}$ using (\ref{eqt:$k^{th}$-order orthogonal condition final estimator 3-replaced}a)-(\ref{eqt:$k^{th}$-order orthogonal condition final estimator 3-replaced}c).}\label{alg:$k^{th}$-order orthogonal condition estimate}
\begin{algorithmic}[1]
\STATE {\bfseries Input:} Observational dataset {\small $\{(y_{m},d_{m},\mathbf{z}_{m})\}_{m=1}^{N} = \mathscr{I}\cup\mathscr{I}^{c}$}, and {\small $\mathscr{I}\cap\mathscr{I}^{c}=\emptyset$}.\label{alg:state 1}
\STATE Train {\small $g^{i}$} and {\small $\pi^{i}$} using the observed data to obtain the estimated nuisance parameters {\small $\hat{g}^{i}$} and {\small $\hat{\pi}^{i}$}.
\STATE For each $i \in \{1,\dots,n\}$: i) relabel the observed data point {\small $(y,d,\mathbf{z})$} as {\small $(y,\tilde{d},\mathbf{z})$} such that {\small $\tilde{d}=1$} if {\small $d=d^{i}$} and {\small $\tilde{d}=0$} if {\small $d\neq d^{i}$}; ii) compute {\small $\tilde{d}-\hat{\pi}^{i}(\mathbf{z})$} for each observation and obtain the local moment {\small $\mathbb{E}\left[(\nu^i)^{q}\mid\mathbf{Z}\right]$} in \eqref{eqt:$k^{th}$-order orthogonal condition score function 2} for each {\small $q$} with the mean of all {\small $(\tilde{d}-\hat{\pi}^{i}(\mathbf{z}))^{q}$}.\label{alg:state 3}
\STATE Compute (\ref{eqt:$k^{th}$-order orthogonal condition final estimator 3-replaced}a)-(\ref{eqt:$k^{th}$-order orthogonal condition final estimator 3-replaced}b) using the observational data.\label{alg:state 4}
\STATE Compute {\small $y-\hat{g}^{i}(\mathbf{z})$} for each observation in {\small $\mathscr{I}$} and store the computed values in {\small $\hat{\mathcal{A}}^{rlz}$} such that {\small $\hat{\mathcal{A}}^{rlz}=\{y_{m}-\hat{g}^{i}(\mathbf{z}_{m})\mid m\in \mathscr{I}\}$}.\label{alg:state 5} 
\STATE For the {\small $m^{\mathrm{th}}$} individual in {\small $\mathscr{I}^c$}, compute {\small $A(d_{m},\mathbf{z}_{m};\hat{\pi}^{i})$}. 
\STATE Repeat a random picking procedure {\small $R$} times: picking an element {\small $\hat{\xi}^{i;rlz}_{m,u}$} randomly in the {\small $u^\mathrm{th}$} repeat for the {\small $m^{\mathrm{th}}$} individual. Then compute {\small $\frac{1}{R}\overset{R}{\underset{u=1}{\sum}}\left[\frac{1}{N}\underset{m\in\mathscr{I}^{c}}{\overset{}{\sum}}\hat{\xi}_{m,u}^{i;rlz} A(d_{m},\mathbf{z}_{m};\hat{\pi}^{i})\right]$} as an estimate of (\ref{eqt:$k^{th}$-order orthogonal condition final estimator 3-replaced}c).\label{alg:state 6}
\STATE {\bfseries Return:} Use the values in Step \ref{alg:state 4} and \ref{alg:state 6} to get the estimate of (\ref{eqt:$k^{th}$-order orthogonal condition final estimator 3-replaced}a)-(\ref{eqt:$k^{th}$-order orthogonal condition final estimator 3-replaced}c).\label{alg:state 8}
\end{algorithmic}
\end{algorithm}

\section{Numerical Studies}\label{sec:experiment}
In this section, we compare the performances of our RCL estimators with the DML estimator and the DR estimator through simulation and empirical experiments. In both experiments, we consider three types of regressors: Lasso, Random Forests (RF), and Multi-layer Perceptron (MLP); and three types of classifiers: Logistic Regression (LR), RF, and MLP. We combine the regression model A and the classification model B to estimate {\small $g^{i}$} and {\small $\pi^i$} respectively, and denote the combination as A+B, e.g., Lasso+LR. In the empirical experiments, we consider two additional state-of-the-art neural network models in causal inference: TARNet \cite{shalit2017estimating} and Dragonnet \cite{shi2019adapting}. All the experiments are run on Dell 3640 with Intel(R) Xeon(R) W-1290P CPU at 3.70GHz, and a set of NVIDIA GeForce RTX 2080Ti GPU.

For all the experiments throughout the paper, we use the following two metrics to evaluate the performance:
\vspace{-0.1cm}
\begin{subequations}
\begin{equation}
{\small
\begin{gathered}
\epsilon_{ATE}=\frac{1}{M}\underset{m=1}{\overset{M}{\sum}}\epsilon_{ATE;m} \; ;
\end{gathered} \label{eqt:metric}
}
\end{equation}
\vspace{-0.2cm}
\begin{equation}
{\small
\begin{gathered}
\sigma_{ATE}=\sqrt{\frac{1}{M-1}\underset{m=1}{\overset{M}{\sum}}\left[\epsilon_{ATE;m}-\epsilon_{ATE}\right]^{2}}.
\end{gathered} \label{eqt:metric variance}
}
\end{equation}
\end{subequations}
Here, {\small $\epsilon_{ATE;m}$} is the weighted relative error of the $m^{\mathrm{th}}$ experiment such that {\small $\epsilon_{ATE;m}=\frac{\underset{\underset{1 \leq i,j \leq n}{i \neq j}}{\overset{}{\sum}}\left|\hat{\theta}^{i,j;m}-\theta^{i,j;m}\right|}{\underset{\underset{1 \leq i,j \leq n}{i \neq j}}{\overset{}{\sum}}\left|\theta^{i,j;m}\right|}$} with {\small $\theta^{i,j;m}$} and {\small $\hat{\theta}^{i,j;m}$} being the true ATE and the estimated ATE between the treatment {\small $d^i$} and the treatment {\small $d^j$} of the {\small $m^{\mathrm{th}}$} experiment. {\small $n$} is the number of treatments and {\small $M$} is the number of experiments.
\subsection{Numerical Studies on Simulation Datasets}\label{sec:Numerical Studies on Simulated Data}
We first introduce the data generating process (DGP) for the simulation experiments. Given the covariates {\small $\mathbf{Z}=(Z_{1},...,Z_{p})^{T}$} which follow a standard multivariate Gaussian distribution, the treatment variable {\small $D$} has the treatment space {\small $\{d^{1}, d^{2}, d^{3}\}$} with the corresponding probability
\begin{equation}\label{eqt:prob_i}
{\small
\begin{aligned}
\pi^{i}(\mathbf{Z}) = \mathbb{P}\{D=d^{i} | \mathbf{Z}\} = \frac{\exp\left(\underset{u=1}{\overset{\lfloor p\cdot r_{c}\rfloor}{\sum}}\beta_{iu}Z_u\right)}{\underset{j=1}{\overset{3}{\sum}}\exp\left(\underset{u=1}{\overset{\lfloor p\cdot r_{c}\rfloor}{\sum}}\beta_{ju}Z_u\right)},
\end{aligned}
}
\end{equation}
where the values of coefficients {\small $\beta_{iu}$} are randomly picked from the uniform distribution {\small $\mathcal{U}(-0.1, 0.1)$}. {\small $r_c$} is the confounding ratio ranging from {\small $0$} to {\small $1$}, and the number of covariates in {\small $\mathbf{Z}$} used to generate {\small $D$} is {\small $p\cdot r_c$} and {\small $\lfloor p\cdot r_{c}\rfloor\in\mathbb{N}$}. For example, if {\small $p=10$} and {\small $r_{c} = 0.56$}, then {\small $p\cdot r_c = 5.6$} and {\small $\lfloor p\cdot r_{c}\rfloor= 5$}. We generate the potential outcome {\small $Y^{i}$} for treatment indices {\small $i \in \{1,2,3\}$} as
\begin{equation}
{\small 
\begin{aligned}
	Y^{i}& = g(d^{i}, \mathbf{Z})+\xi^{i} = e^{\sqrt{d^{i}}}\left(\bm{a}_{i}^{T}\mathbf{Z}+1\right)^2+\xi^{i},
\end{aligned}
}
\end{equation}
where {\small $\bm{a}_{i}$} is a {\small $p\times 1$} constant vector whose elements are randomly chosen from {\small $\mathcal{U}(0.1, 0.5)$}. We also set {\small $d^{1}=0.1$}, {\small $d^{2}=0.5$}, {\small $d^{3}=1$}, {\small $\xi^{1} \sim \mathcal{N}(0,9)$}, {\small$\xi^{2} \sim \mathcal{N}(0,4)$} and {\small$\xi^{3} \sim \mathcal{N}(0,1)$}. Next, we generate {\small $N$} i.i.d. observations based on the DGP. Suppose the realized covariates of the {\small $m^{\mathrm{th}}$} individual are {\small $\mathbf{z}_{m}$}, then the actual treatment {\small $d_{m}$} will be $d^{k}$, where $k$ is determined by {\small $k=\underset{u \in \{1,2,3\}}{\arg\max} \; \pi^{u}(\mathbf{z}_{m})$}. Under the actual treatment $d^{k}$, the observed factual outcome {\small $y_{m}$} will correspondingly be {\small $y^{k}$}.

For the simulation experiments, we compute the DR, DML and our RCL estimators with different values of {\small $r$} and {\small $k$} (see Theorem \ref{thm:$k^{th}$-order orthogonal condition score function}), which is denoted by RCL{\small $_{r,k}$}. We then use {\small $\epsilon_{ATE}$} in \eqref{eqt:metric} with {\small $n=3$} and {\small $M=100$} to evaluate the performance of different estimators for each combination of the regressor {\small $g^{i}$} and the classifier {\small $\pi^{i}$} (denoted as regressor+classifier). We split every dataset by the ratio {\small $56\%/14\%/30\%$} as training/validation/test sets.
\paragraph{Consistency of RCL estimators}  In this part, we set {\small $r_c=1$}, {\small $p=5$}, and let the number of observations {\small $N$} vary in {\small $\{1, 2, 4, 8, 16\}\times 10000$}. We check the consistency of RCL estimators through simulations and report {\small $\epsilon_{ATE}$} in Fig. \ref{Figure:ATE_convergency_r=4}. The result indicates that the error reduces when the sample size increases for our RCL estimators. Besides, we also find that when {\small $g^{i}$} is fitted well {\small $\forall i$} (e.g., when the regressor is chosen as Lasso or RF), RCL{\small $_{2,2}$} performs better than DR, DML and other RCL estimators. On the other hand, when {\small $g^{i}$} is not fitted well for some {\small $i$}, e.g., when the regressor is chosen as MLP, the DML and the RCL estimator with $k=1$ can significantly correct the bias thanks to the doubly robust property. In this case, despite similar performances produced by RCL{\small $_{2,1}$} estimator and the DML estimator, RCL{\small $_{2,1}$} still has a smaller {\small $\epsilon_{ATE}$}.
\begin{figure}[ht]
	\setlength{\belowcaptionskip}{-0.3cm}
	\centering	
	\includegraphics[width=1\columnwidth]{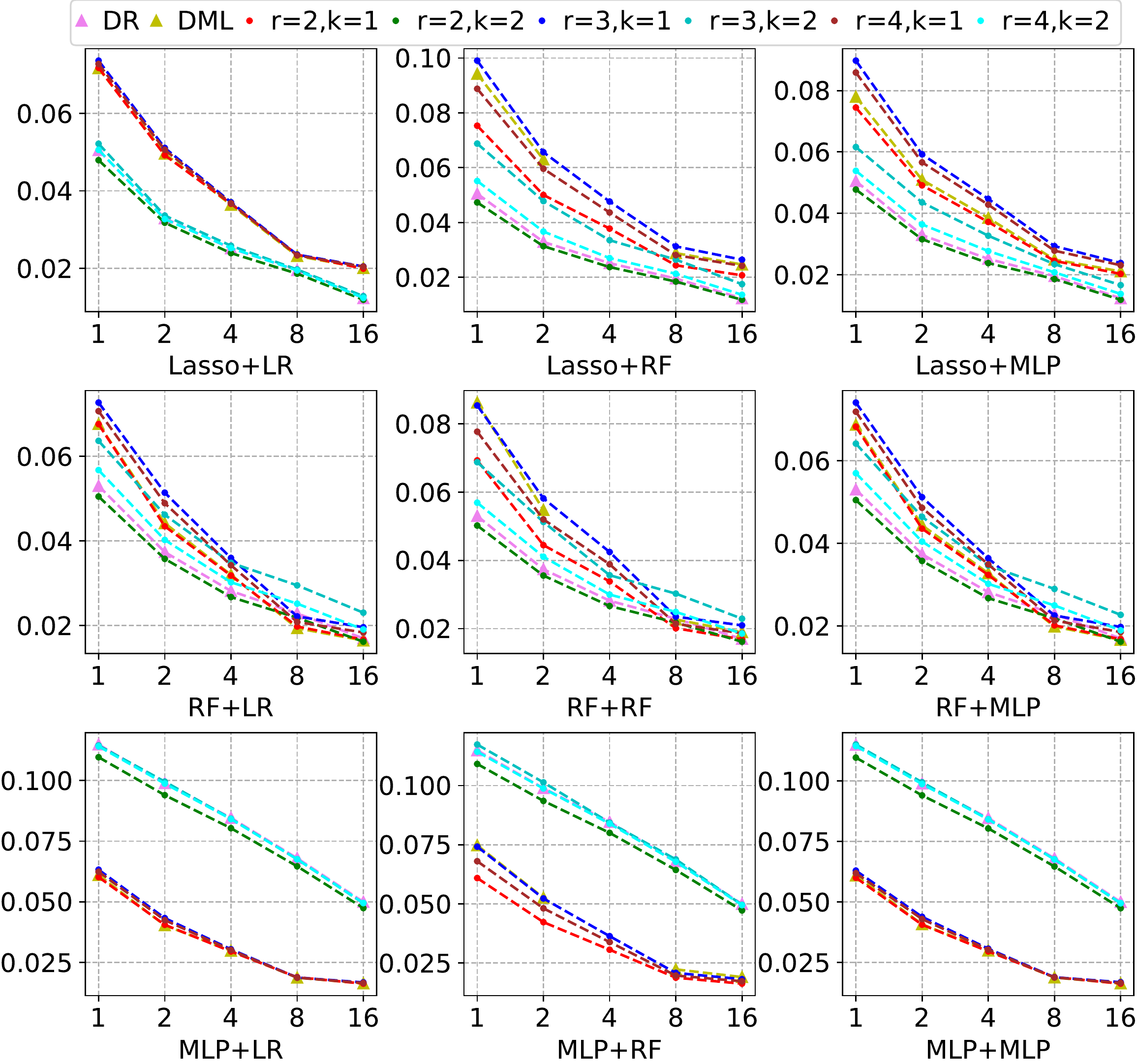}
	\caption{Plots of {\small {\small $\epsilon_{ATE}$}} versus the varying {\small $N$}: DR vs DML vs RCL.}
	\label{Figure:ATE_convergency_r=4}
\end{figure}
\begin{figure}[ht]
\setlength{\belowcaptionskip}{-0.3cm}
\centering
\includegraphics[width=0.9\columnwidth]{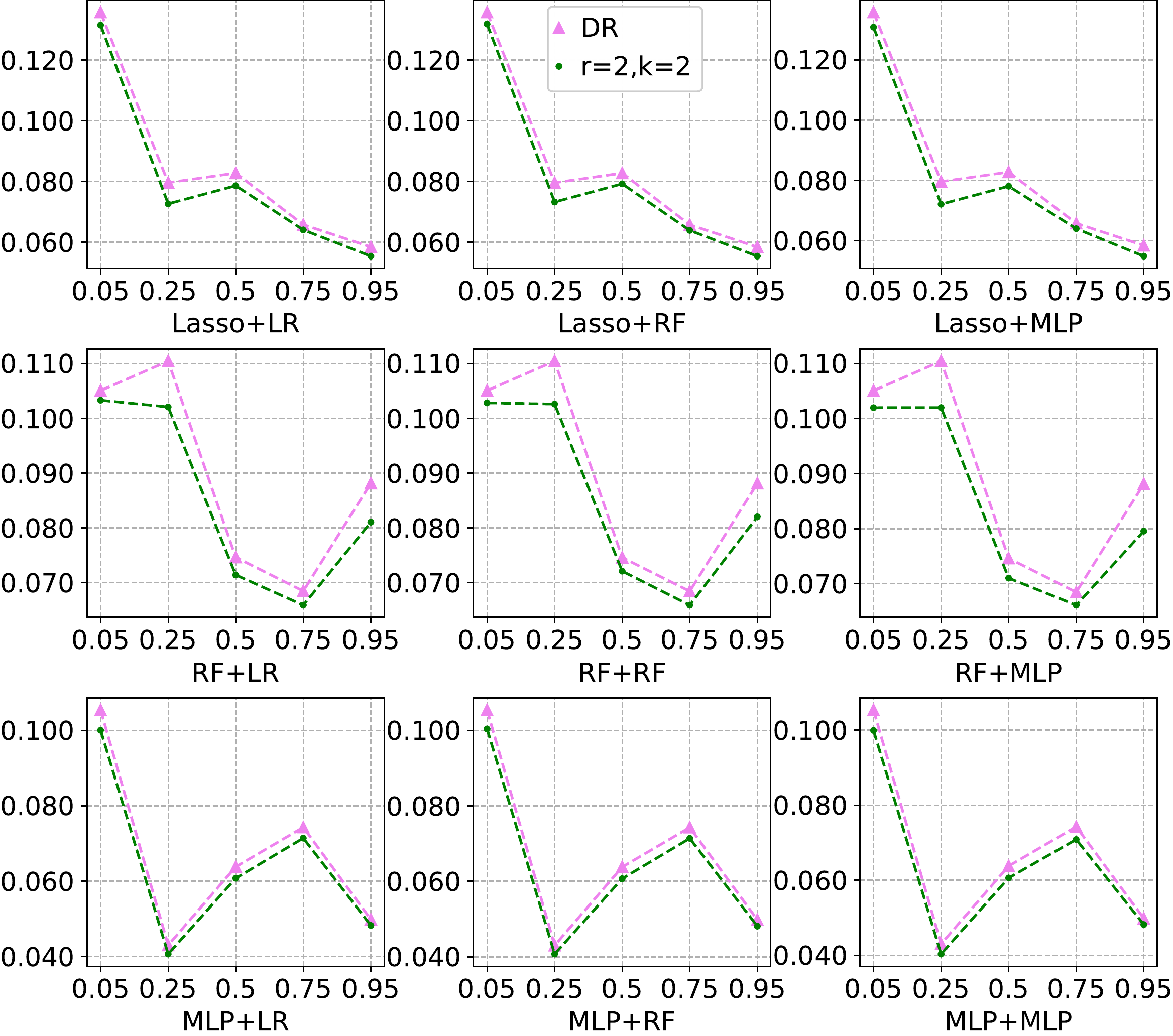}
\caption{Plots of {\small $\epsilon_{ATE}$} versus the varying {\small $r_c$}: DR vs RCL{\small $_{2,2}$}.}
\label{Figure:confounding_ratio_DRvsROL}
\end{figure}
\begin{figure}[ht]
\setlength{\belowcaptionskip}{-0.3cm}
\centering
\includegraphics[width=0.9\columnwidth]{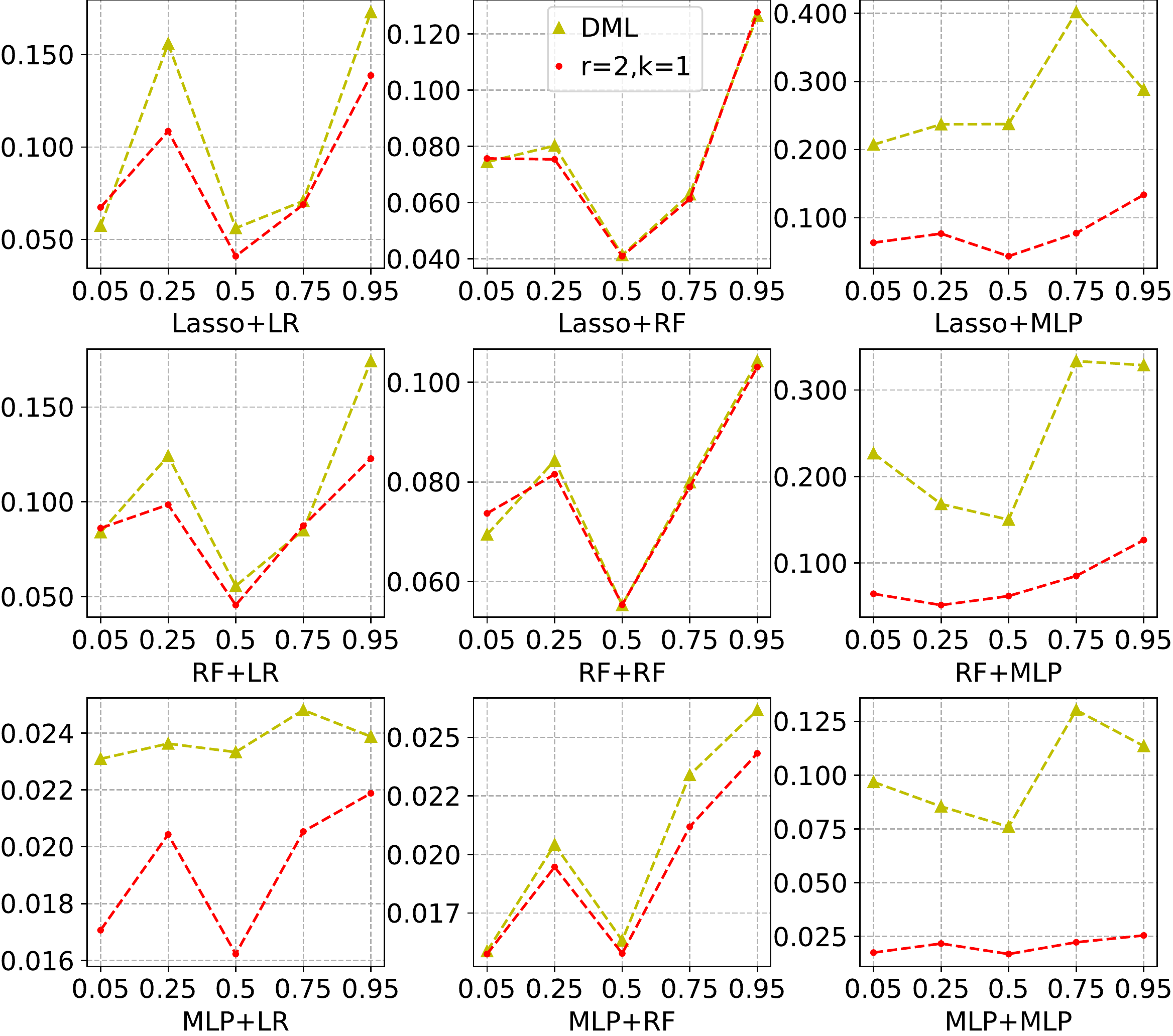}
\caption{Plots of {\small $\epsilon_{ATE}$} versus the varying {\small $r_c$}: DML vs RCL{\small $_{2,1}$}.}
\label{Figure:confounding_ratio_DMLvsROL}
\end{figure}
\begin{figure}[ht]
\setlength{\belowcaptionskip}{-0.3cm}
\centering
\includegraphics[width=0.9\columnwidth]{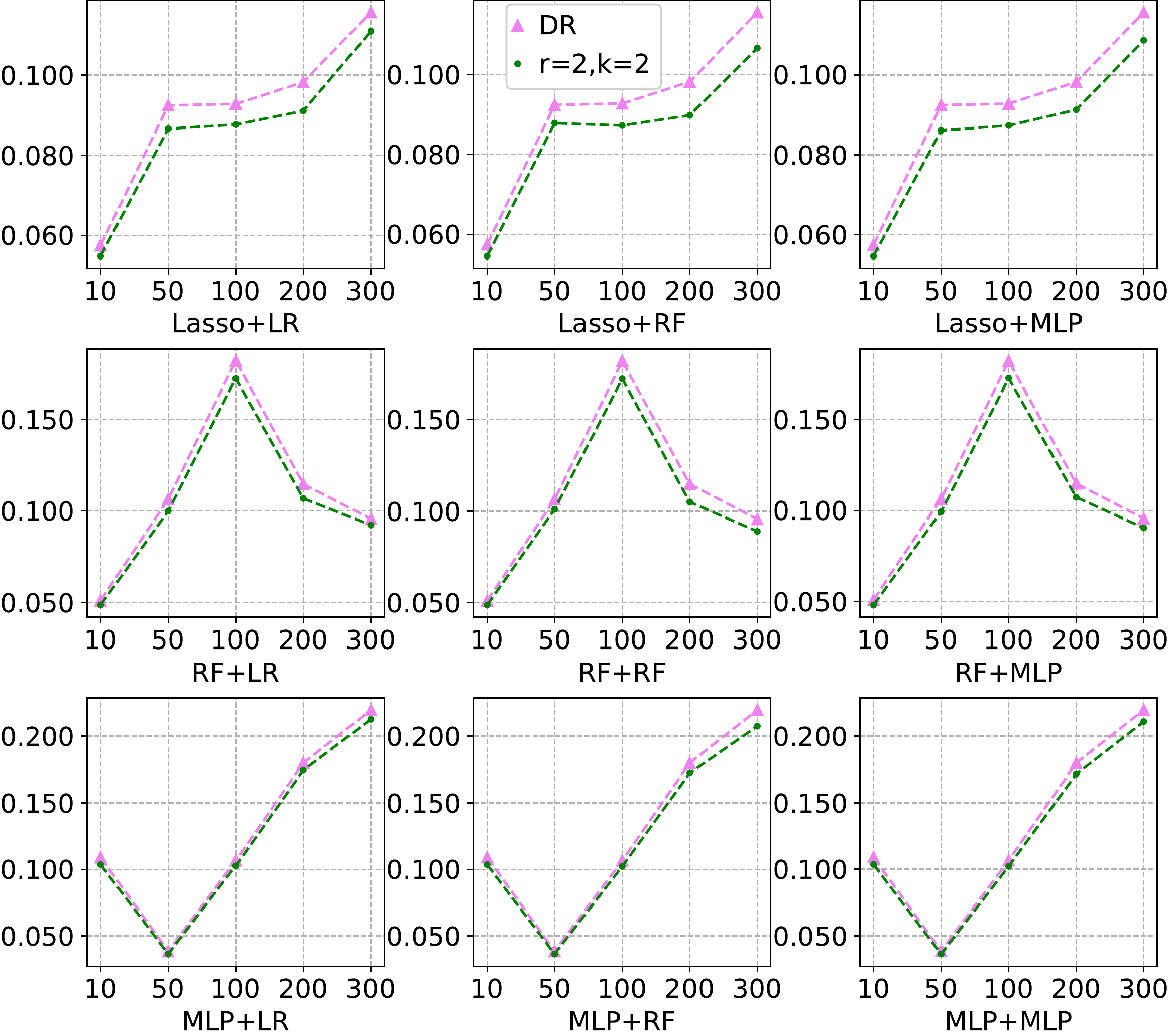}
\caption{Plots of {\small $\epsilon_{ATE}$} versus the varying {\small $p$}: DR vs RCL{\small $_{2,2}$}.}
\label{Figure:dim_DRvsROL}
\end{figure}
\begin{figure}[ht]
\setlength{\belowcaptionskip}{-0.3cm}
\centering
\includegraphics[width=0.9\columnwidth]{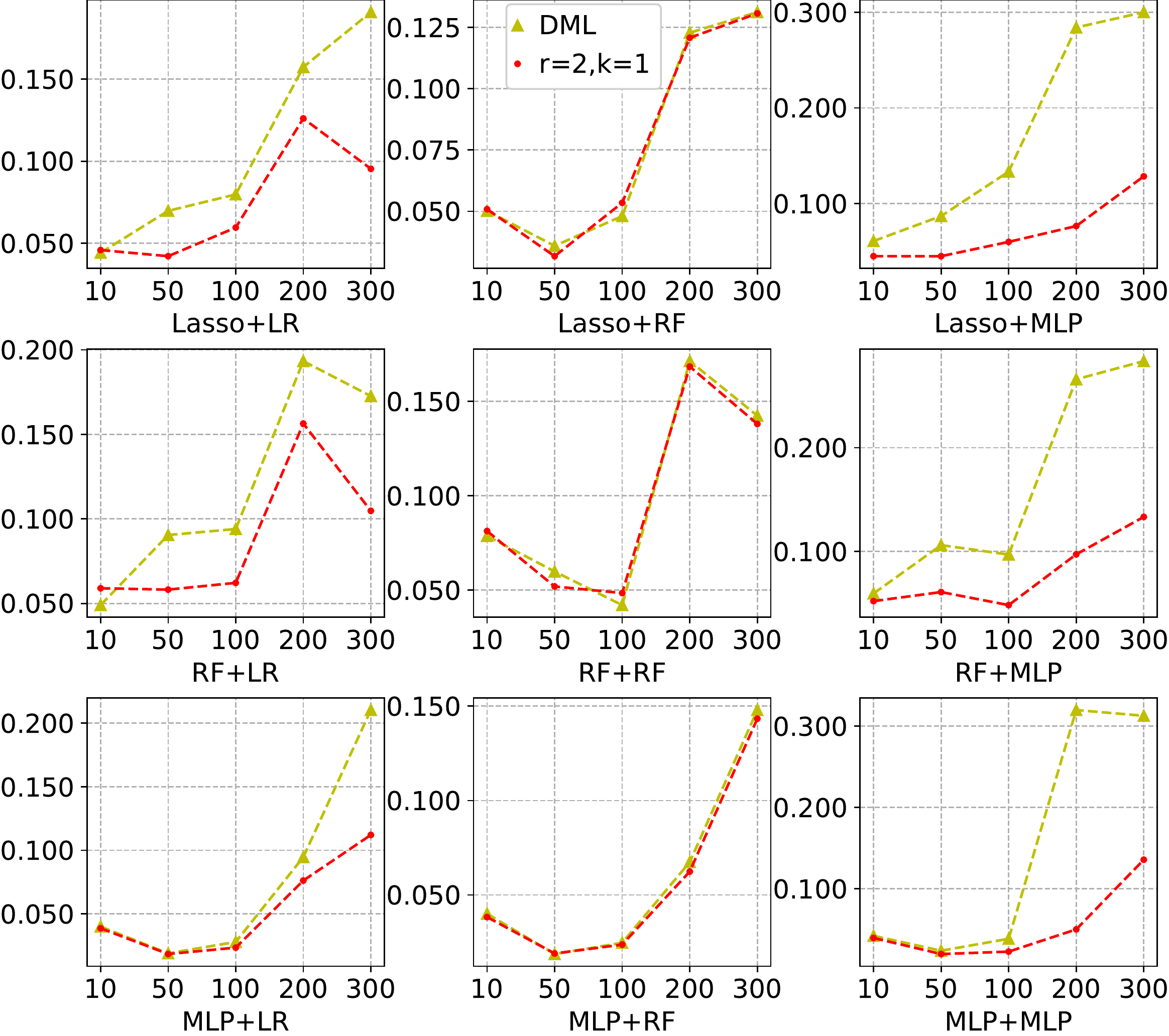}
\caption{Plots of {\small $\epsilon_{ATE}$} versus the varying {\small $p$}: DML vs RCL{\small $_{2,1}$}.}
\label{Figure:dim_DMLvsROL}
\end{figure}
\paragraph{Varying {\small $r_{c}$} and {\small $p$}} In the following experiments, we mainly compare RCL{\small $_{2,2}$} with DR, and RCL{\small $_{2,1}$} with DML since the above simulation experiments indicate that RCL{\small $_{2,2}$} and DR perform similarly, and RCL{\small $_{2,1}$} has similar trends to DML. We set {\small $N=10000$} and plot {\small $\epsilon_{ATE}$} produced by each model combination A+B versus i) different {\small $r_{c}$} with {\small $p = 100$} in Fig. \ref{Figure:confounding_ratio_DRvsROL} and Fig. \ref{Figure:confounding_ratio_DMLvsROL}; ii) different {\small $p$} with {\small $r_{c}=1$} in Fig. \ref{Figure:dim_DRvsROL} and Fig. \ref{Figure:dim_DMLvsROL}. From all the four figures, we observe that the DML estimator is sensitive to the change of {\small $r_{c}$} and {\small $p$}, especially when the classifier is MLP. As analyzed before, if the estimation error of {\small $\pi^{i}(\mathbf{z})$} is non-negligible for some $\mathbf{z}$, the term {\small $\frac{1}{\hat{\pi}^i(\cdot)}$} of the DML estimator often gives extreme values especially when {\small$\hat{\pi}^{i}(\mathbf{z})$} is small, leading to pronounced estimation errors of the ATE. Indeed, any ATE estimators that involve the inverse propensity score term might face this error-compounding issue. By contrast, our RCL estimators are less volatile to the variation of {\small $r_{c}$} and {\small $p$} regardless of the choice of classifiers. For example, in Fig. \ref{Figure:dim_DMLvsROL}, we notice that when the classifier is MLP, the error of DML rises dramatically as {\small $p$} increases, while RCL{\small $_{2,1}$} performs more steadily. In addition, our RCL{\small $_{2,2}$} estimator overall has a smaller {\small $\epsilon_{ATE}$} than the DR estimator no matter how {\small $r_{c}$} or {\small $p$} varies.
\subsection{Numerical Studies on Benchmark Datasets}\label{sec:Numerical Studies on Benchmark Datasets}
\begin{table*}[htbp]
	\centering
	\caption{The performance comparisons ({\small $\epsilon_{ATE} \pm \sigma_{ATE}$}) on the test sets of 1000 \textbf{IHDP} experiments. Smaller {\small $\epsilon_{ATE}$} is better.}
	\resizebox{2\columnwidth}{!}{
		\begin{tabular}{cccc|cccccc}
			\toprule
			Model/Estimator & DR    & RCL$_{2,2}$ & $R_{DR}$ & IPW   & AIPW  & DML   & DML-trim & RCL$_{2,1}$ & $R_{DML}$ \\
			\midrule
			LASSO+LR & \num{0.0957016592574597}$\pm$\num{0.640267801173915} & \num{0.0922338093233415}$\pm$\num{0.577391472448973} & -3.6\% & \num{1.51697356015605}$\pm$\num{4.68598260312737} & \num{0.105763367167141}$\pm$\num{0.300194083971267} & \num{0.105763367167142}$\pm$\num{0.300194083971274} & \num{0.105763364379096}$\pm$\num{0.300194084194018} & \num{0.104851958331204}$\pm$\num{0.651543659112282} & -0.9\% \\
			\midrule
			LASSO+RF & \num{0.0957016592574597}$\pm$\num{0.640267801173915} & \num{0.0919418871606084}$\pm$\num{0.564172928180374} & -3.9\% & $\infty$ & \num{0.109315513248708}$\pm$\num{0.280292014664635} & $\infty$ & \num{0.11131590918084}$\pm$\num{0.291181180228278} & \num{0.10262849397032}$\pm$\num{0.630798484505913} & -6.1\% \\
			\midrule
			LASSO+MLP & \num{0.0957016592574597}$\pm$\num{0.640267801173915} & \num{0.0927439264967579}$\pm$\num{0.593285296206187} & -3.1\% & \num{4.4812490222709}$\pm$\num{9.85644139084187} & \num{0.224440190445999}$\pm$\num{0.439720169984648} & \num{0.224440190445998}$\pm$\num{0.439720169984648} & \num{0.19311183721892}$\pm$\num{0.323636374244496} & \num{0.102944659987956}$\pm$\num{0.590220193929016} & -47\% \\
			\midrule
			RF+LR & \num{0.0979732966141551}$\pm$\num{0.553484192555486} & \num{0.0932579681403953}$\pm$\num{0.439817975315819} & -4.8\% & \num{1.51697356015605}$\pm$\num{4.68598260312737} & \num{0.128811621589012}$\pm$\num{0.724585718194816} & \num{0.128811621589012}$\pm$\num{0.724585718194816} & \num{0.128811615676305}$\pm$\num{0.724585718507295} & \num{0.128842821124574}$\pm$\num{1.07488187681154} & 0.0\% \\
			\midrule
			RF+RF & \num{0.0979732966141551}$\pm$\num{0.553484192555486} & \num{0.0942932571289985}$\pm$\num{0.458204054417463} & -3.8\% & $\infty$ & \num{0.126178712647616}$\pm$\num{0.638434948404866} & $\infty$ & \num{0.128012877019247}$\pm$\num{0.641268017780346} & \num{0.124621883086877}$\pm$\num{1.01221269744655} & -1.2\% \\
			\midrule
			RF+MLP & \num{0.0979732966141551}$\pm$\num{0.553484192555486} & \num{0.0950802814635367}$\pm$\num{0.506860794218963} & -3.0\% & \num{4.4812490222709}$\pm$\num{9.85644139084187} & \num{0.240480141381655}$\pm$\num{0.756511151700561} & \num{0.240480141381655}$\pm$\num{0.756511151700565} & \num{0.211482597532412}$\pm$\num{0.714411783365695} & \num{0.120590099634166}$\pm$\num{0.883571527750748} & -43\% \\
			\midrule
			MLP+LR & \num{0.0803597106682557}$\pm$\num{0.229579144869116} & \num{0.0790235103312214}$\pm$\num{0.229472621327462} & -1.7\% & \num{1.51697356015605}$\pm$\num{4.68598260312737} & \num{0.141389283589274}$\pm$\num{0.242181566583048} & \num{0.141389283589274}$\pm$\num{0.242181566583048} & \num{0.14138927504712}$\pm$\num{0.242181564569142} & \num{0.104578167591645}$\pm$\num{0.194193628272569} & -26\% \\
			\midrule
			MLP+RF & \num{0.0803597106682557}$\pm$\num{0.229579144869116} & \num{0.0792146226254808}$\pm$\num{0.23321971636867} & -1.4\% & $\infty$ & \num{0.140930737288581}$\pm$\num{0.244936557918274} & $\infty$ & \num{0.142145985158482}$\pm$\num{0.253716565669282} & \num{0.104382624931009}$\pm$\num{0.192382428814327} & -26\% \\
			\midrule
			MLP+MLP & \num{0.0803597106682557}$\pm$\num{0.229579144869116} & \num{0.0790382014019174}$\pm$\num{0.230289370683907} & -1.6\% & \num{4.4812490222709}$\pm$\num{9.85644139084187} & \num{0.388796191412648}$\pm$\num{1.11801289407659} & \num{0.388796144909208}$\pm$\num{1.11801240474788} & \num{0.340278452668023}$\pm$\num{0.951125010966888} & \num{0.111926344793105}$\pm$\num{0.235810638754905} & -67\% \\
			\midrule
			TARNet & \num{0.0542958025852373}$\pm$\num{0.0942507273076831} & \num{0.0533817393804141}$\pm$\num{0.0918769235152564} & -1.7\% & \num{1.27631249863221}$\pm$\num{4.1830027255246} & \num{0.0888336751009626}$\pm$\num{0.16250044416974} & \num{0.0888336437847067}$\pm$\num{0.162499466513899} & \num{0.0888336437847067}$\pm$\num{0.162499466513899} & \num{0.0827653474603322}$\pm$\num{0.197713461439485} & -6.8\% \\
			\midrule
			Dragonnet & \num{0.0562447309384025}$\pm$\num{0.092229254884046} & \num{0.0557504174151612}$\pm$\num{0.0952493269346849} & -0.9\% & \num{1.71597732398186}$\pm$\num{4.81418695573581} & \num{0.14127353877066}$\pm$\num{0.218639997523957} & \num{0.141273601602494}$\pm$\num{0.218640208253127} & \num{0.134702721565834}$\pm$\num{0.175778675670746} & \num{0.0816722418781141}$\pm$\num{0.105565159717243} & -39\% \\
			\bottomrule
		\end{tabular}%
	}
	\label{Table:ihdp}%
\end{table*}%
\begin{table*}[htbp]
	\centering
	\caption{The performance comparisons ({\small $\epsilon_{ATE} \pm \sigma_{ATE}$}) on the test sets of 100 \textbf{Twins} experiments. Smaller {\small $\epsilon_{ATE}$} is better.}
	\resizebox{2\columnwidth}{!}{
		\begin{tabular}{cccc|cccccc}
			\toprule
			Model/Estimator & DR    & RCL$_{2,2}$ & $R_{DR}$ & IPW   & AIPW  & DML   & DML-trim & RCL$_{2,1}$ & $R_{DML}$ \\
			\midrule
			LASSO+LR & \num{0.667161417975565}$\pm$\num{0.436943790888348} & \num{0.646031944788284}$\pm$\num{0.357810609635846} & -3.2\% & \num{0.997186654546707}$\pm$\num{1.0282479839264} & \num{0.863413170171111}$\pm$\num{0.867510368220328} & \num{0.863413170171112}$\pm$\num{0.867510368220329} & \num{0.863413170171112}$\pm$\num{0.867510368220329} & \num{0.860539637381624}$\pm$\num{0.861111151796158} & -0.3\% \\
			\midrule
			LASSO+RF & \num{0.667161417975565}$\pm$\num{0.436943790888348} & \num{0.640423458872207}$\pm$\num{0.346180238288157} & -4.0\% & \num{1.06774689631041}$\pm$\num{1.06667674961721} & \num{0.912542741177314}$\pm$\num{1.0324487837971} & \num{0.912542741177313}$\pm$\num{1.03244878379709} & \num{0.912541530559238}$\pm$\num{1.03244888082876} & \num{0.850930835731041}$\pm$\num{0.857651696249796} & -6.8\% \\
			\midrule
			LASSO+MLP & \num{0.667161417975565}$\pm$\num{0.436943790888348} & \num{0.652263208887379}$\pm$\num{0.361908866758876} & -2.2\% & \num{2.71695259473322}$\pm$\num{2.53524699605525} & \num{0.941193658246498}$\pm$\num{0.969680994312327} & \num{0.941193658246499}$\pm$\num{0.969680994312327} & \num{0.924807382693669}$\pm$\num{0.950933057524706} & \num{0.883284524889292}$\pm$\num{0.851073779438786} & -4.5\% \\
			\midrule
			RF+LR & \num{0.604186638495974}$\pm$\num{0.533120559180304} & \num{0.576393675529002}$\pm$\num{0.460608409580317} & -4.6\% & \num{0.997186654546707}$\pm$\num{1.0282479839264} & \num{0.782618974579244}$\pm$\num{0.863295302774404} & \num{0.782618974579244}$\pm$\num{0.863295302774404} & \num{0.782618974579244}$\pm$\num{0.863295302774404} & \num{0.774451947731414}$\pm$\num{0.842482685707567} & -1.0\% \\
			\midrule
			RF+RF & \num{0.604186638495974}$\pm$\num{0.533120559180304} & \num{0.573526624706884}$\pm$\num{0.444114614484407} & -5.1\% & \num{1.06774689631041}$\pm$\num{1.06667674961721} & \num{0.898143875599452}$\pm$\num{1.04811245118099} & \num{0.898143875599452}$\pm$\num{1.04811245118099} & \num{0.898142228357805}$\pm$\num{1.04811284592583} & \num{0.808639013083756}$\pm$\num{0.860979316949181} & -10\% \\
			\midrule
			RF+MLP & \num{0.604186638495974}$\pm$\num{0.533120559180304} & \num{0.582058275375756}$\pm$\num{0.467041836393486} & -3.7\% & \num{2.71695259473322}$\pm$\num{2.53524699605525} & \num{0.878704531177447}$\pm$\num{0.939717406064745} & \num{0.878704531177446}$\pm$\num{0.939717406064746} & \num{0.863391078384253}$\pm$\num{0.925297128540437} & \num{0.816129289317978}$\pm$\num{0.828611304610332} & -5.5\% \\
			\midrule
			MLP+LR & \num{0.660213653748403}$\pm$\num{0.643159826971596} & \num{0.624421249388836}$\pm$\num{0.561873977426024} & -5.4\% & \num{0.997186654546707}$\pm$\num{1.0282479839264} & \num{0.822347815959208}$\pm$\num{0.831239772507} & \num{0.822347815959207}$\pm$\num{0.831239772507} & \num{0.822347815959207}$\pm$\num{0.831239772507} & \num{0.817012134151291}$\pm$\num{0.827875248022536} & -0.6\% \\
			\midrule
			MLP+RF & \num{0.660213653748403}$\pm$\num{0.643159826971596} & \num{0.618517938691176}$\pm$\num{0.548179257855073} & -6.3\% & \num{1.06774689631041}$\pm$\num{1.06667674961721} & \num{0.939531974914234}$\pm$\num{1.00613870291003} & \num{0.939531974914234}$\pm$\num{1.00613870291003} & \num{0.939531328888188}$\pm$\num{1.00613877340099} & \num{0.84527502195103}$\pm$\num{0.827317719691973} & -10\% \\
			\midrule
			MLP+MLP & \num{0.660213653748403}$\pm$\num{0.643159826971596} & \num{0.629679914590112}$\pm$\num{0.562780660147385} & -4.6\% & \num{2.71695259473322}$\pm$\num{2.53524699605525} & \num{0.904776867847457}$\pm$\num{0.954993073594451} & \num{0.904776832924196}$\pm$\num{0.95499336557208} & \num{0.899217557920781}$\pm$\num{0.953636899261841} & \num{0.849724793635334}$\pm$\num{0.833052100933173} & -5.5\% \\
			\midrule
			TARNet & \num{0.656097905805353}$\pm$\num{0.600217082337399} & \num{0.621069087185189}$\pm$\num{0.509983143308584} & -5.3\% & \num{2.46916252879544}$\pm$\num{2.98410306499364} & \num{0.937674888447282}$\pm$\num{1.30252582133794} & \num{0.937675061452057}$\pm$\num{1.30252599182589} & \num{0.937675061452057}$\pm$\num{1.30252599182589} & \num{0.864644737634207}$\pm$\num{1.04331803418323} & -7.8\% \\
			\midrule
			Dragonnet & \num{0.677258215220963}$\pm$\num{0.634940870296755} & \num{0.642171326254504}$\pm$\num{0.561434959073984} & -5.2\% & \num{1.66850913554347}$\pm$\num{1.65120065182147} & \num{0.794897423744431}$\pm$\num{0.740918933191085} & \num{0.794897437630666}$\pm$\num{0.740918949177126} & \num{0.794897437630666}$\pm$\num{0.740918949177126} & \num{0.79024433520866}$\pm$\num{0.773968125856086} & -0.6\% \\
			\bottomrule
		\end{tabular}%
	}
	\label{Table:twins}%
\end{table*}%
\paragraph{Models} Similar to the simulation experiments, we choose Lasso, RF, and MLP as the regressors while LR, RF, and MLP as the classifiers. Additionally, two prevalent neural network models, TARNet and Dragonnet, are also considered for learning the nuisance parameters. According to \cite{shi2019adapting}, these two neural network structures can incorporate the estimations of both {\small $g^{i}$} and {\small $\pi^{i}$} using the representation learning technique.

\paragraph{Settings} We implement the above methods on two widely adopted benchmark datasets for causal inference, i.e., \textbf{IHDP} and \textbf{Twins}, and then compare RCL estimators with DR, IPW, DML, and their variants AIPW and DML-trim estimators. Mathematically, both the AIPW estimator and the DML-trim estimator are the same as the DML estimator. However, empirically, AIPW and DML-trim are less prone to suffer the extreme values. To be precise, AIPW decomposes the estimator into two parts that both contain the IPW term (see \cite{linden2016estimating}), while DML-trim trims estimated propensity scores at the cutoff points of $0.01$ and $0.99$ (see \cite{chernozhukov2018double}).

We take the RCL{\small $_{2,1}$} and RCL{\small $_{2,2}$} as the representatives of the general RCL estimators because for the real datasets with a relatively small sample size and a large dimension of features, the second-moment estimation of {\small $\nu^i$} is more reliable compared to the higher-moment estimations. We use grid search to adjust the hyperparameters for those general machine learning models. For TARNet and Dragonnet, we use the same network structures (layers, units, regularization, batch size, learning rate, and stopping criterion) as suggested in \cite{shalit2017estimating} and \cite{shi2019adapting}.

\paragraph{IHDP} It is a widely used benchmark dataset for causal inference introduced by \cite{hill2011bayesian}. IHDP dataset is constructed based on the randomized controlled experiment conducted by Infant Health and Development Program. The collected 25-dimensional confounders from the 747 samples are associated with the properties of infants and their mothers, such as birth weight and mother's age. Our aim is to study the treatment effect of the specialist visits (binary treatment) on the cognitive scores (continuous-valued outcome). By removing a subset of the treated group, the selection bias in the IHDP dataset occurs. 
There are 1000 IHDP datasets given in \cite{hill2011bayesian}. Each dataset is split by the ratio of {\small $63\%/27\%/10\%$} as training/validation/test sets, which keeps consistent with \cite{shalit2017estimating}.

\paragraph{Twins} Twins dataset is introduced by \cite{louizos2017causal} and it collects twin births in the USA between 1989 and 1991. The treatment {\small $D=1$} indicates the heavier twin while {\small $D=0$} indicates the lighter twin; the outcome {\small $Y$} is a binary variable defined as the mortality in the first year; the covariates {\small $\mathbf{Z}$} include 30 features relevant to the parents, the pregnancy and the birth. Similar to \cite{yoon2018ganite}, we only select twins that have the same gender and both weigh less than 2kg. Finally, we have {\small $11440$} pairs of twins whose mortality rates are {\small $17.7\%$} for lighter twin and {\small $16.1\%$} for heavier twin. To simulate an observational dataset with selection bias, we selectively choose one of the two twins as the observed sample based on the covariates of $m^{\text{th}}$ individual: {\small $D_m|\mathbf{Z}_m \sim$} Bernoulli(Sigmoid({\small $\mathbf{w}^T\mathbf{Z}_m+b$})), where {\small $\mathbf{w}^T \sim \mathcal{U}_{30}((-0.01,0.01)^{30})$} and {\small $b \sim \mathcal{N}(0,0.01)$}. We repeat this process {\small $100$} times, and each of the generated {\small $100$} Twins datasets is split by the ratio of {\small $64\%/16\%/20\%$} as training/validation/test sets, which keeps consistent with \cite{yoon2018ganite}.

\paragraph{Analysis} In Table \ref{Table:ihdp} and Table \ref{Table:twins}, we report the performance of every model combination, measured by {\small $\epsilon_{ATE} \; (\pm \sigma_{ATE})$}, for IHDP and Twins experiments, respectively. The smaller {\small $\epsilon_{ATE}$}, the better. The metric {\small $R_{DR}=\text{RCL}_{2,2}/\text{DR}-1$} ({\small $R_{DML}=\text{RCL}_{2,1}/\min(\text{IPW, AIPW, DML, DML-trim})-1$}) is used to evaluate the reduction ratio in {\small $\epsilon_{ATE}$} of RCL{\small $_{2,2}$} relative to DR (RCL{\small $_{2,1}$} relative to the best estimator among IPW, AIPW, DML, and DML-trim). The negative {\small $R_{DR}$} ({\small $R_{DML}$}) indicates that the RCL estimator has a smaller {\small $\epsilon_{ATE}$} than the DR (IPW, AIPW, DML, and DML-trim) estimator.

Table \ref{Table:ihdp} reports the experimental results on IHDP datasets. It illustrates that although the DR estimator produces reasonable estimates, the RCL{\small $_{2,2}$} estimator has a more minor {\small $\epsilon_{ATE}$} than the DR estimator, with the error reduced relatively by {\small $0.9\%-4.8\%$}. Simultaneously, RCL{\small $_{2,2}$} achieves the best performance among all the estimators across all the model combinations. We also notice that even though the AIPW and DML-trim avoid extreme values encountered by DML (e.g., when the classifier is chosen as RF, the inverse propensity score is estimated with an infinity value for some data points), the RCL{\small $_{2,1}$} estimator is still at most $67\%$ better than the best of IPW, DML, and DML-trim estimators. More importantly, when the variance of inverse propensity scores is large (e.g., when the classifier is MLP), the improvement of RCL{\small $_{2,1}$} to DML becomes more substantial.

Table \ref{Table:twins} presents the experimental results on Twins datasets. It can be observed that the RCL{\small $_{2,2}$} estimator has a significantly smaller {\small $\epsilon_{ATE}$} compared with other estimators for all model combinations, and it can reduce the estimation error relatively by {\small $2.2\%-6.3\%$} compared with the DR method. Besides, the RCL{\small $_{2,1}$} estimator can reduce the estimation error by {\small $0.3\%-10\%$} relative to the best of IPW, DML, and DML-trim estimators. It is also noticeable that when {\small $\pi^{i}$} is well specified (e.g., the case on using Dragonnet in Table \ref{Table:twins}), our RCL{\small $_{2,1}$} estimator still outperforms the DML estimator even though the error {\small $\epsilon_{ATE}$} produced by the DML estimator is small enough.

\subsection{Numerical Studies on Credit dataset}
Causal inference benchmark datasets are typically generated by a parametric data generating process. Though the ground truth of treatment effects are accessable in this way, such semi-synthetic datasets fail to resemble the original real data sets.

In summary, DML is recognized as a better method than DR and IPW because when the DR (IPW) estimator has a notable bias due to the misspecification on {\small $g^{i}$} ({\small $\pi^{i}$}), the DML estimator can reduce the bias if {\small $\pi^{i}$} ({\small $g^{i}$}) is well estimated. However, the advantages of DML are not easy to achieve in practice. First, the DML estimator, which incorporates the inverse propensity score term, may give a very large estimation or even infinite value of the ATE, reflecting that the DML estimator is volatile to the estimation of propensity scores. Second, if {\small $g^{i}$} is approximated well enough, DML will not assuredly perform better than DR due to the high variance of the IPW term. By contrast, our RCL estimators are more practical since i) they can stabilize the error caused by the misspecification on propensity scores; ii) if {\small $g^{i}$} is well approximated, the RCL{\small $_{2,2}$} estimator will outperform the DR estimator owing to the RCL scores are orthogonal scores; iii) if {\small $g^{i}$} is not well approximated, but {\small $\pi^{i}$} is correctly specified, the RCL estimator with {\small $k=1$} performs better than the DML estimator with smaller estimation errors and slighter volatility to the estimated propensity scores.

\section{Conclusion}
This paper constructs the RCL scores and establishes the RCL estimators for the ATE estimation. Theoretically, we prove that the RCL scores are orthogonal scores and the RCL estimators are consistent. Numerically, the comprehensive experiments have shown that our estimators outperform the commonly used estimators such as DR, IPW, AIPW, DML, and DML-trim estimators. In addition, the proposed RCL estimators have the same merit, i.e., the doubly robust property, as the DML estimator. However, unlike the DML estimator, the RCL estimators are more stable to the estimation error due to the misspecification on propensity scores than the DML estimator and its variants. In the future research, we will i) investigate the optimal values of {\small $(r,k)$} in Theorem \ref{thm:$k^{th}$-order orthogonal condition score function}; and ii) provide interpretability for deep learning models in causal inference using the RCL method.

\section{Acknowledgements}
Qi WU acknowledges the support from the Hong Kong Research Grants Council [General Research Fund 14206117, 11219420, and 11200219], CityU SRG-Fd fund 7005300, and the support from the CityU-JD Digits Laboratory in Financial Technology and Engineering, HK Institute of Data Science. The work described in this paper was partially supported by the InnoHK initiative, The Government of the HKSAR, and the Laboratory for AI-Powered Financial Technologies.

Shumin MA acknowledges the support from the Guangdong Provincial Key Laboratory of Interdisciplinary Research and Application for Data Science, BNU-HKBU United International College under project code 2022B1212010006, the support from Guangdong Higher Education Upgrading Plan (2021-2025) of ``Rushing to the Top, Making Up Shortcomings and Strengthening Special Features" with UIC research grant R0400001-22, and the UIC grant UICR0700019-22.

\section{Appendices}\label{sec:Appendices}
\subsection{proofs}\label{sec:proofs}
We present the theoretical proofs of Theorems and Corollaries given in the paper.
\begin{proof}[\textbf{Proof of Theorem \ref{thm:$k^{th}$-order orthogonal condition score function}}]\ \\
Given the nuisance parameters $\varrho=(\mathcal{g}^{i},a_{i})$ and the true nuisance parameters $\rho=(g^{i},\pi^{i})$, we find out the RCL score $\psi^{i}(W,\vartheta,\varrho)$ w.r.t. the nuisance parameters $\varrho=(\mathcal{g}^{i},a_{i})$ which can be used to construct the estimators of the causal parameter $\theta^{i}:=\mathbb{E}\left[g^{i}(\mathbf{Z})\right]$. We try an ansatz of $\psi^{i}(W,\vartheta,\varrho)$ such that
\begin{equation}
{\small
\begin{aligned}\label{eqt:higher-order orthogonal condition score function-supp}
\psi^{i}(W,\vartheta,\varrho)&=\vartheta-\mathcal{g}^{i}(\mathbf{Z})-(Y^{i}-\mathcal{g}^{i}(\mathbf{Z}))A(D,\mathbf{Z};a_{i}),
\end{aligned}
}
\end{equation}
where 
\begin{equation}
{\small
\begin{aligned}\label{eqt:higher-order orthogonal condition score function 2-supp}
A(D,\mathbf{Z};a_{i})&=\bar{b}_{r}\left[\mathbf{1}_{\{D=d^{i}\}}-a_{i}(\mathbf{Z})\right]^{r}\\
&+\underset{q=1}{\overset{k-1}{\sum}}b_{q}\big(\big[\mathbf{1}_{\{D=d^{i}\}}-a_{i}(\mathbf{Z})\big]^{q}-\mathbb{E}\big[(\nu^{i})^{q}\mid\mathbf{Z}\big]\big).
\end{aligned}
}
\end{equation}
Here, the coefficients $b_{1},\cdots,b_{k-1},\bar{b}_{r}$ depend on $\mathbf{Z}$ and $\nu^{i}$ only. Using the ansatz, we notice that $\psi^{i}(W,\vartheta,\varrho)$ satisfies the moment condition, i.e., $\mathbb{E}\left[\psi^{i}(W,\vartheta,\varrho)\mid_{\vartheta=\theta^{i},\;\varrho=\rho}\right]=0$. Indeed, we have
\begin{equation*}
{\small
\begin{aligned}
&\mathbb{E}\left[\psi^{i}(W,\vartheta,\varrho)\mid_{\vartheta=\theta^{i},\;\varrho=\rho}\right]\\
=&\mathbb{E}\left[\theta^{i}-g^{i}(\mathbf{Z})-(Y^{i}-g^{i}(\mathbf{Z}))A(D,\mathbf{Z};\pi^{i})\right]\\
=&\mathbb{E}\left[\theta^{i}-g^{i}(\mathbf{Z})\right]-\mathbb{E}\left[(Y^{i}-g^{i}(\mathbf{Z}))A(D,\mathbf{Z};\pi^{i})\right]\\
=&-\mathbb{E}\left[\xi^{i}\times A(D,\mathbf{Z};\pi^{i})\right]\\
=&-\mathbb{E}\left[\mathbb{E}\left[\xi^{i}\times A(D,\mathbf{Z};\pi^{i})\mid D, \mathbf{Z}\right]\right]\\
=&-\mathbb{E}\left[A(D,\mathbf{Z};\pi^{i})\mathbb{E}\left[\xi^{i}\mid D,\mathbf{Z}\right]\right]=0.\\
\end{aligned}
}
\end{equation*}
The second last equality comes from the fact that $A(D,\mathbf{Z};\pi^{i})$ is a function of $(D,\mathbf{Z})$. The last equality comes from the fact that $(\xi^{i} \perp \!\!\! \perp D) \mid \mathbf{Z}$. Now, we aim to find out the coefficients of $b_{1},\cdots,b_{k-1},\bar{b}_{r}$ such that the score \eqref{eqt:higher-order orthogonal condition score function-supp} satisfies the $k^{\mathrm{th}}$ score. Indeed, we need to have $\mathbb{E}\left[\partial_{\mathcal{g}^{i}}^{\alpha_{1}}\partial_{a_{i}}^{\alpha_{2}}\psi^{i}(W,\vartheta,\varrho)\mid_{\vartheta=\theta^{i},\;\varrho=\rho}\mid\mathbf{Z}\right]=0$ for all $\alpha_{1}$ and $\alpha_{2}$ which are non-negative integers such that $1\leq \alpha_{1}+\alpha_{2}\leq k$. Since $\partial_{\mathcal{g}^{i}}^{\alpha_{1}}\partial_{a_{i}}^{\alpha_{2}}\psi^{i}(W,\vartheta,\varrho)=0$ when $\alpha_{1}\geq 2$, we only need to solve the coefficients $b_{1},\cdots,b_{k-1},\bar{b}_{r}$ from
\begin{subequations}
{\small
\begin{empheq}[left=\empheqlbrace]{align}
0&=\mathbb{E}\left[\partial_{a_{i}}^{k}\psi^{i}(W,\vartheta,\varrho)\mid_{\vartheta=\theta^{i},\;\varrho=\rho}\mid\mathbf{Z}\right],\label{eqt:score condition 1}\\
0&=\mathbb{E}\left[\partial_{\mathcal{g}^{i}}^{1}\partial_{a_{i}}^{q}\psi^{i}(W,\vartheta,\varrho)\mid_{\vartheta=\theta^{i},\;\varrho=\rho}\mid\mathbf{Z}\right]\label{eqt:score condition 2}
\end{empheq}
}\noindent
\end{subequations}
$\forall q=0,\cdots,k-1$. However, \eqref{eqt:score condition 1} always holds since
\begin{equation*}
{\small
\begin{aligned}
&\mathbb{E}\left[\partial_{a_{i}}^{k}\psi^{i}(W,\vartheta,\varrho)\mid_{\vartheta=\theta^{i},\;\varrho=\rho}\mid\mathbf{Z}\right]\\
=&\mathbb{E}\left[(Y^{i}-g^{i}(\mathbf{Z}))\times\partial_{a_{i}}^{k}A(D,\mathbf{Z};a_{i})\mid_{a_{i}=\pi^{i}}\mid\mathbf{Z}\right]\\
=&\mathbb{E}\left[\mathbb{E}\left[(Y^{i}-g^{i}(\mathbf{Z}))\times\partial_{a_{i}}^{k}A(D,\mathbf{Z};a_{i})\mid_{a_{i}=\pi^{i}}\mid D,\mathbf{Z}\right] \mid \mathbf{Z}\right]\\										=&\mathbb{E}\left[\partial_{a_{i}}^{k}A(D,\mathbf{Z};a_{i})\mid_{a_{i}=\pi^{i}}\mathbb{E}\left[(Y^{i}-g^{i}(\mathbf{Z}))\mid D,\mathbf{Z}\right] \mid \mathbf{Z}\right]\\
=&\mathbb{E}\left[\partial_{a_{i}}^{k}A(D,\mathbf{Z};a_{i})\mid_{a_{i}=\pi^{i}}\mathbb{E}\left[\xi^{i}\mid \mathbf{Z}\right] \mid D,\mathbf{Z}\right]=0.
\end{aligned}
}
\end{equation*}
Consequently, we need to find out the coefficients $b_{1}$, $b_{2}$, $\dots$, $b_{k-1}$, $\bar{b}_{r}$ from
{\small
\begin{align}\tag{\ref{eqt:score condition 2}}
\mathbb{E}\left[\partial_{\mathcal{g}^{i}}^{1}\partial_{a_{i}}^{q}\psi^{i}(W,\vartheta,\varrho)\mid_{\vartheta=\theta^{i},\;\varrho=\rho}\mid\mathbf{Z}\right]&=0
\end{align}
}\noindent
$\forall q=0,\cdots,k-1$. From \eqref{eqt:score condition 2}, there are $k$ equations and we need to solve the $k$ unknowns $b_{1},\cdots,b_{k-1},\bar{b}_{r}$ from the $k$ equations. Generally, the $k$ unknowns could be solved uniquely.

To start with, we compute $\partial_{\mathcal{g}^{i}}^{1}\partial_{a_{i}}^{q}\psi^{i}(W,\vartheta,\varrho)$ for $q=0,\cdots,k-1$. Note that
\begin{equation*}
{\small
\begin{aligned}
\partial_{\mathcal{g}^{i}}^{1}\partial_{a_{i}}^{q}\psi^{i}(W,\vartheta,\varrho)&=-1+A(D,\mathbf{Z};a_{i})
\end{aligned}
}
\end{equation*}
when $q=0$ and
\begin{equation*}
{\small
\begin{aligned}
\partial_{\mathcal{g}^{i}}^{1}\partial_{a_{i}}^{q}\psi^{i}(W,\vartheta,\varrho)&=\bar{b}_{r}\frac{r!(-1)^{q}[\mathbf{1}_{\{D=d^{i}\}}-a_{i}(\mathbf{Z})]^{r-q}}{(r-q)!}\\
&\quad+\overset{k-1}{\underset{u=q}{\sum}}b_{u}\frac{u!(-1)^{q}[\mathbf{1}_{\{D=d^{i}\}}-a_{i}(\mathbf{Z})]^{u-q}}{(u-q)!}
\end{aligned}
}
\end{equation*}
when $1\leq q\leq k-1$. Consequently, we need to solve for $b_{1},\cdots,b_{k-1}$ and $\bar{b}_{r}$ simultaneously from
\begin{subequations}
\begin{equation}
{\small
\begin{aligned}\label{eqt:simultaneous equation 1}
1&=\mathbb{E}\left[A(D,\mathbf{Z};\pi^{i})\mid\mathbf{Z}\right]
\end{aligned}
}
\end{equation}
and 
\begin{equation}
{\small
\begin{aligned}\label{eqt:simultaneous equation 2}
0&=\mathbb{E}\left[\bar{b}_{r}\frac{r!(-1)^{q}[\mathbf{1}_{\{D=d^{i}\}}-\pi^{i}(\mathbf{Z})]^{r-q}}{(r-q)!}\mid\mathbf{Z}\right]\\
&\quad+\mathbb{E}\left[\overset{k-1}{\underset{u=q}{\sum}}b_{u}\frac{u!(-1)^{q}[\mathbf{1}_{\{D=d^{i}\}}-\pi^{i}(\mathbf{Z})]^{u-q}}{(u-q)!}\mid\mathbf{Z}\right].
\end{aligned}
}
\end{equation}
\end{subequations}
From \eqref{eqt:simultaneous equation 1}, we have
\begin{equation}\tag{\ref{eqt:simultaneous equation 1}*}\label{eqt:simultaneous equation 1-simplified}
{\small
\begin{aligned}
1&=\bar{b}_{r}\mathbb{E}\left[\left(\mathbf{1}_{\{D=d^{i}\}}-\pi^{i}(\mathbf{Z})\right)^{r}\mid\mathbf{Z}\right]\\
&\quad+\underset{q=1}{\overset{k-1}{\sum}}b_{q}\mathbb{E}\left[\left(\mathbf{1}_{\{D=d^{i}\}}-\pi^{i}(\mathbf{Z})\right)^{q}\mid\mathbf{Z}\right]\\
&\quad-\underset{q=1}{\overset{k-1}{\sum}}b_{q}\mathbb{E}\left[\mathbb{E}\left[(\nu^{i})^{q}\mid\mathbf{Z}\right]\mid\mathbf{Z}\right].
\end{aligned}
}
\end{equation}
Since 
\begin{equation*}
\begin{aligned}
\mathbb{E}\left[\left(\mathbf{1}_{\{D=d^{i}\}}-\pi^{i}(\mathbf{Z})\right)^{q}\mid\mathbf{Z}\right]&=\mathbb{E}\left[(\nu^{i})^{q}\mid\mathbf{Z}\right]\quad\text{and}\\
\mathbb{E}\left[\mathbb{E}\left[(\nu^{i})^{q}\mid\mathbf{Z}\right]\mid\mathbf{Z}\right]&=\mathbb{E}\left[(\nu^{i})^{q}\mid\mathbf{Z}\right],
\end{aligned}
\end{equation*}
we understand that 
\begin{equation*}
\begin{aligned}
\mathbb{E}\left[\left(\mathbf{1}_{\{D=d^{i}\}}-\pi^{i}(\mathbf{Z})\right)^{q}\mid\mathbf{Z}\right]-\mathbb{E}\left[\mathbb{E}\left[(\nu^{i})^{q}\mid\mathbf{Z}\right]\mid\mathbf{Z}\right]=0.
\end{aligned}
\end{equation*}
As such, \eqref{eqt:simultaneous equation 1-simplified} can be reduced as
%
\begin{equation*}
{\small
\begin{aligned}
\bar{b}_{r}\mathbb{E}\left[\left(\mathbf{1}_{\{D=d^{i}\}}-\pi^{i}(\mathbf{Z})\right)^{r}\mid\mathbf{Z}\right]&=1\\
\Rightarrow \bar{b}_{r}\mathbb{E}\left[(\nu^{i})^{r}\mid\mathbf{Z}\right]&=1.
\end{aligned}
}
\end{equation*}
Hence, we can solve for $\bar{b}_{r}$ such that
\begin{equation*}
\begin{aligned}
\bar{b}_{r}=\frac{1}{\mathbb{E}\left[(\nu^{i})^{r}\mid\mathbf{Z}\right]}.
\end{aligned}
\end{equation*}
It remains to find out $b_{1},\cdots,b_{k-1}$ from \eqref{eqt:simultaneous equation 2}. Indeed, we can simplify \eqref{eqt:simultaneous equation 2} as
\begin{equation}
{\small
\begin{aligned}\label{eqt:iteration equation of finding coeff of b}
\bar{b}_{r}\mathbb{E}\left[\frac{r!(\nu^{i})^{r-q}}{(r-q)!}\mid\mathbf{Z}\right]+\overset{k-1}{\underset{u=q}{\sum}}b_{u}\mathbb{E}\left[\frac{u!(\nu^{i})^{u-q}}{(u-q)!}\mid\mathbf{Z}\right] &=0\\
\Rightarrow \bar{b}_{r}\binom{r}{q}\mathbb{E}\left[(\nu^{i})^{r-q}\mid\mathbf{Z}\right]+\overset{k-1}{\underset{u=q}{\sum}}b_{u}\binom{u}{q}\mathbb{E}\left[(\nu^{i})^{u-q}\mid\mathbf{Z}\right] &=0
\end{aligned}
}
\end{equation}
$\forall 1\leq q\leq k-1$. Now, we solve $b_{1},\cdots,b_{k-1}$. We start with finding out $b_{k-1}$, followed by $b_{k-2},\;b_{k-3},\cdots,b_{1}$ iteratively. When $q=k-1$, \eqref{eqt:iteration equation of finding coeff of b} becomes
\begin{equation*}
{\small
\begin{aligned}
0&=\bar{b}_{r}\binom{r}{k-1}\mathbb{E}\left[(\nu^{i})^{r-k+1}\mid\mathbf{Z}\right]\\
&\qquad\qquad+b_{k-1}\binom{k-1}{k-1}\mathbb{E}\left[(\nu^{i})^{0}\mid\mathbf{Z}\right]\\
\Rightarrow b_{k-1}&=-\bar{b}_{r}\binom{r}{k-1}\mathbb{E}\left[(\nu^{i})^{r-k+1}\mid\mathbf{Z}\right].
\end{aligned}
}
\end{equation*}
Now, when $q=k-2$, \eqref{eqt:iteration equation of finding coeff of b} becomes
\begin{equation*}
{\small
\begin{aligned}
0&=\bar{b}_{r}{r \choose k-2}\mathbb{E}\left[(\nu^{i})^{r-k+2}\mid\mathbf{Z}\right]\\
&\quad+b_{k-1}{k-1 \choose k-2}\mathbb{E}\left[(\nu^{i})^{(k-1)-(k-2)}\mid\mathbf{Z}\right]\\
&\quad+b_{k-2}\mathbb{E}\left[(\nu^{i})^{0}\mid\mathbf{Z}\right]\\
\Rightarrow b_{k-2}&=-b_{k-1}{k-1 \choose k-2}\mathbb{E}\left[(\nu^{i})^{1}\mid\mathbf{Z}\right]\\
&\quad-\bar{b}_{r}{r \choose k-2}\mathbb{E}\left[(\nu^{i})^{r-k+2}\mid\mathbf{Z}\right].
\end{aligned}
}
\end{equation*}
Now, suppose $b_{q+1},\cdots,b_{k-1}$ are known and we want to find out what $b_{q}$ is. We have to solve it from
\begin{equation*}
{\small
\begin{aligned}
0&=b_{q}\mathbb{E}\left[(\nu^{i})^{0}\mid\mathbf{Z}\right]\\
&\quad\quad+\bar{b}_{r}{r \choose q}\mathbb{E}\left[(\nu^{i})^{r-q}\mid\mathbf{Z}\right]+\overset{k-1}{\underset{u=q+1}{\sum}}b_{u}{u \choose q}\mathbb{E}\left[(\nu^{i})^{u-q}\mid\mathbf{Z}\right].
\end{aligned}
}
\end{equation*}
We can obtain $b_{q}$ from the above equation, which gives
\begin{equation*}
{\small
\begin{aligned}
b_{q}&=-\overset{k-1}{\underset{u=q+1}{\sum}}b_{u}{u \choose q}\mathbb{E}\left[(\nu^{i})^{u-q}\mid\mathbf{Z}\right]\\
&\quad\quad-\bar{b}_{r}{r \choose q}\mathbb{E}\left[(\nu^{i})^{r-q}\mid\mathbf{Z}\right]\\
\Rightarrow b_{q}&=-\overset{k-1-q}{\underset{u=1}{\sum}}b_{q+u}{q+u \choose q}\mathbb{E}\left[(\nu^{i})^{u}\mid\mathbf{Z}\right]\\
&\quad\quad-\bar{b}_{r}{r \choose q}\mathbb{E}\left[(\nu^{i})^{r-q}\mid\mathbf{Z}\right].
\end{aligned}
}
\end{equation*}
The proof is completed.
\end{proof}

\begin{proof}[\textbf{Proof of Corollary \ref{corollary:$k^{th}$-order orthogonal condition estimator}}]\ 
We have discussed the way to obtain the estimator in the main paper.
\end{proof}
To facilitate the upcoming studies, we first introduce some notations. Recall that $Y^{i}$ is the potential outcome under the treatment $d^{i}$. We use $Y^{i;F}$ as the factual outcome if an individual receives $d^i$, and $Y^{i;CF}$ as the counterfactual outcome if an individual receives alternative treatments. Hence, we have
\begin{equation*}
\begin{aligned}
Y^{i}=\begin{cases}
Y^{i;F} & \text{if $D=d^{i}$}\\
Y^{i;CF} & \text{if $D\neq d^{i}$}
\end{cases}.
\end{aligned}
\end{equation*}
Based on the introduced notations, we can define two \textit{residual differences} $\xi_{m}^{i;F}$ and $\xi_{m}^{i;CF}$ for the $m^{\text{th}}$ individual according to the sets $\mathscr{I}$ and $\mathscr{I}^{c}$ that the $m^{\text{th}}$ individual belongs to. Mathematically, $\xi_m^{i;F}:=Y_m^{i;F}-g^{i}(\mathbf{Z}_m)$ if $m \in \mathscr{I}$ and $\xi_m^{i;CF}:=Y_m^{i;CF}-g^{i}(\mathbf{Z}_m)$ if $m \in \mathscr{I}^c$. 

We give the statistical properties between $\xi_{m}^{i;F}$ and $\xi_{\bar{m}}^{i;CF}$ for $m,\;\bar{m}$.  First, $\xi_{m}^{i;F}\perp\!\!\!\perp\xi_{\bar{m}}^{i;CF}\mid \mathbf{Z}$ due to the SUTVA assumption for $m\neq\bar{m}$. Second. Regardless of the actual treatment the individual receives, the noises in terms of $d^{i}$ should be identical and independently distributed for different individuals. As a result, we should have
\begin{equation*}
\begin{aligned}
\xi_{m}^{i;F}\overset{d}{=}\xi_{\bar{m}}^{i;CF}\mid \mathbf{Z}\quad\text{and}\quad\xi_{m}^{i;F}\perp\!\!\!\perp\xi_{\bar{m}}^{i;CF}\mid \mathbf{Z}.
\end{aligned}
\end{equation*}
From the above assumptions, we have $\mathbb{E}\big[(\xi_{m}^{i;F})^{r}\mid\mathbf{Z}\big]=\mathbb{E}\big[(\xi_{\bar{m}}^{i;CF})^{r} \mid\mathbf{Z}\big]$ $\forall\;r$ and $m,\bar{m}$.

We give the statistical properties between $\xi_{m}^{i;F}$ and $\xi_{\bar{m}}^{i;CF}$. To be precise, we study $\xi_{m}^{i;F}$ and $\xi_{\bar{m}}^{i;CF}$ conditioning on $\mathbf{Z}$. The properties are summarized in Proposition \ref{xi proposition}.
\begin{proposition} \label{xi proposition}
Given the covariates $\mathbf{Z}$, the random variable $\xi_{m}^{i;F}$ and $\xi_{\bar{m}}^{i;CF}$ are independent and identically distributed, i.e., $\mathbb{E}\big[(\xi_{m}^{i;F})^{r}\mid D_{m},\mathbf{Z}\big]=\mathbb{E}\big[(\xi_{\bar{m}}^{i;CF})^{r} \mid D_{\bar{m}},\mathbf{Z}\big]$ and $\xi_{m}^{i;F} \perp\!\!\!\perp \xi_{\bar{m}}^{i;CF}\mid \mathbf{Z}$.
\end{proposition}
\begin{proof}
Using the SUTVA assumption, we have $\xi_{m}^{i;F} \perp\!\!\!\perp \xi_{\bar{m}}^{i;CF}\mid\mathbf{Z}$. In addition, we have
\begin{equation*}
\resizebox{0.5\textwidth}{!}{$
\begin{aligned}
&\mathbb{E}\left[(\xi_{m}^{i;F})^{r} \mid D_{m}, \mathbf{Z}_{m}\right]=\mathbb{E}\left[(Y_{m}^{i;F}-g^{i}(\mathbf{Z}_{m}))^{r} \mid D_{m},\mathbf{Z}_{m}\right]\\
=&\underset{k=0}{\overset{r}{\sum}}\binom{r}{k}\mathbb{E}\left[(Y_{m}^{i;F})^{r} \mid D_{m}, \mathbf{Z}_{m}\right](-g^{i}(\mathbf{Z}_{m}))^{r-k}\\
\overset{\star}{=}&\underset{k=0}{\overset{r}{\sum}}\binom{r}{k}\mathbb{E}\left[(Y_{m}^{i;F})^{r} \mid \mathbf{Z}_{m}\right](-g^{i}(\mathbf{Z}_{m}))^{r-k}=\mathbb{E}\left[(\xi^i)^{r} \mid \mathbf{Z}\right].
\end{aligned}
$}
\end{equation*}
$\overset{\star}{=}$ is due to the ignorability assumption. Similarly, we also have
\begin{equation*}
\resizebox{0.5\textwidth}{!}{$
\begin{aligned}
\mathbb{E}\left[(\xi_{\bar{m}}^{i;CF})^{r} \mid D_{\bar{m}} \neq d^i, \mathbf{Z}_{\bar{m}}\right]=\mathbb{E}\left[(\xi_{\bar{m}}^i)^{r} \mid \mathbf{Z}_{\bar{m}}\right]=\mathbb{E}\left[(\xi^i)^{r} \mid \mathbf{Z}\right].
\end{aligned}
$}
\end{equation*}
\end{proof}
In the remaining sequel, we investigate the consistency of our RCL estimators based on the basics of orthogonal machine learning theory. To start with, we give the assumptions on the nuisance parameters. Only the assumptions that are helpful in studying the consistency of our RCL estimators are stated. Other assumptions that concentrate on the conditions of the scores, including orthogonality, identifiability, non-degeneracy, smoothness, and the regularity of moments can be found in \cite{mackey2018orthogonal} and references therein.
\begin{assumption}\label{RCL assumption}
Given that the nuisance parameters and the true nuisance parameters are $(\hat{g}^{i},\hat{\pi}^{i})$ and $(g^{i},\pi^{i})$, and $S=\{\bm{\alpha}=(\alpha_{1},\alpha_{2})\in\mathbb{Z}^{2}_{\geq 0}:\left\|\bm{\alpha}\right\|_{1}\leq k\}$, we have
\begin{enumerate}
\item {\small $\mathbb{E}\big[\left|\hat{g}^{i}(\mathbf{Z})-g^{i}(\mathbf{Z})\right|^{4\alpha_{1}}\left|\hat{\pi}^{i}(\mathbf{Z})-\pi^{i}(\mathbf{Z})\right|^{4\alpha_{2}}\mid \hat{g}^{i},\hat{\pi}^{i}\big]\overset{p}{\rightarrow} 0$} {\small $\forall\;\bm{\alpha}\in S$}\\
\item {\small $N^{\frac{1}{2}}\sqrt{\mathbb{E}\big[\left|\hat{g}^{i}(\mathbf{Z})-g^{i}(\mathbf{Z})\right|^{2\alpha_{1}}\left|\hat{\pi}^{i}(\mathbf{Z})-\pi^{i}(\mathbf{Z})\right|^{2\alpha_{2}}\mid \hat{g}^{i},\hat{\pi}^{i}\big]}\overset{p}{\rightarrow} 0$} {\small $\forall\;\bm{\alpha}\in \{\bm{\alpha}\in\mathbb{Z}^{2}_{\geq 0}:\left\|\bm{\alpha}\right\|_{1}\leq k+1\}\backslash S$}.
\end{enumerate}
\end{assumption}

For notational convenience, we rewrite our RCL estimator $\hat{\theta}^{i}_{RCL}$ and denote it as $\hat{\theta}^{i}_{N}$. Indeed, we have
\begin{equation}
{\small
\begin{aligned}
\hat{\theta}^{i}_{N}=&\underbrace{\frac{1}{N}\underset{m=1}{\overset{N}{\sum}}\hat{g}^{i}(\mathbf{Z}_{m})}_{(a)}+\underbrace{\frac{1}{N}\underset{m\in \mathscr{I}}{\overset{}{\sum}}(Y^{i;F}_{m}-\hat{g}^{i}(\mathbf{Z}_{m}))\hat{A}_{m}^{i}}_{(b)}\\
&+\underbrace{\frac{1}{R}\overset{R}{\underset{u=1}{\sum}}\left[\frac{1}{N}\underset{m\in \mathscr{I}^{c}}{\overset{}{\sum}}\hat{\xi}_{m,u}^{i;F} \hat{A}_{m}^{i}\right]}_{(c)}.
\end{aligned}
}
\end{equation}
Besides, we define two quantities $\hat{\tilde{\theta}}_{N}^{i}$ and $\hat{\bar{\theta}}_{N}^{i}$. They are
{\small
\begin{align}
\hat{\tilde{\theta}}_{N}^{i}&=\frac{1}{N}\underset{m=1}{\overset{N}{\sum}}\hat{g}^{i}(\mathbf{Z}_{m})+\frac{1}{N}\underset{m\in \mathscr{I}}{\overset{}{\sum}}(Y^{i;F}_{m}-\hat{g}^{i}(\mathbf{Z}_{m}))\hat{A}_{m}^{i}\nonumber\\
&\quad+\frac{1}{N}\underset{m\in \mathscr{I}^{c}}{\overset{}{\sum}}\hat{\xi}^{i;CF}_{m}\hat{A}_{m}^{i},\label{eqt:higher-order orthogonal condition before estimator in appendix}\\
\hat{\bar{\theta}}_{N}^{i}&=\frac{1}{N}\underset{m=1}{\overset{N}{\sum}}\hat{g}^{i}(\mathbf{Z}_{m})+\frac{1}{N}\underset{m\in \mathscr{I}}{\overset{}{\sum}}(Y^{i;F}_{m}-\hat{g}^{i}(\mathbf{Z}_{m}))\hat{A}_{m}^{i}\nonumber\\
&\quad+\frac{1}{N}\underset{m\in \mathscr{I}^{c}}{\overset{}{\sum}}\hat{\xi}_{m}^{i;F}\hat{A}_{m}^{i}.\label{eqt:higher-order orthogonal condition final estimator in appendix}
\end{align}
}\noindent
We also define
\begin{equation*}
\resizebox{0.49\textwidth}{!}{$
\begin{aligned}
\kappa_{N}^{i;F}&=\frac{1}{N}\underset{m\in \mathscr{I}^{c}}{\overset{}{\sum}}\xi_{m}^{i;F}A_{m}^{i},\;\hat{\kappa}_{N}^{i;F}=\frac{1}{N}\underset{m\in \mathscr{I}^{c}}{\overset{}{\sum}}\hat{\xi}_{m}^{i;F}\hat{A}_{m}^{i},\\
\kappa_{N}^{i;CF}&=\frac{1}{N}\underset{m\in \mathscr{I}^{c}}{\overset{}{\sum}}\xi_{m}^{i;CF}A_{m}^{i},\;\hat{\kappa}_{N}^{i;CF}=\frac{1}{N}\underset{m\in \mathscr{I}^{c}}{\overset{}{\sum}}\hat{\xi}_{m}^{i;CF}\hat{A}_{m}^{i},\\
\hat{\kappa}_{R,N}^{i;F}&=\frac{1}{R}\underset{u=1}{\overset{R}{\sum}}\big[\frac{1}{N}\underset{m\in \mathscr{I}^{c}}{\overset{}{\sum}}\hat{\xi}_{m,u}^{i;F}\hat{A}_{m}^{i}\big],\;\kappa_{R,N}^{i;F}=\frac{1}{R}\overset{R}{\underset{u=1}{\sum}} \big[\frac{1}{N}\underset{m\in \mathscr{I}^{c}}{\overset{}{\sum}}\xi_{m,u}^{i;F} A_{m}^{i}\big].
\end{aligned}
$}
\end{equation*}
Then \eqref{eqt:higher-order orthogonal condition before estimator in appendix} and \eqref{eqt:higher-order orthogonal condition final estimator in appendix} can be rewritten as
{\small
\begin{align}
\hat{\tilde{\theta}}_{N}^{i}&=\frac{1}{N}\underset{m=1}{\overset{N}{\sum}}\hat{g}^{i}(\mathbf{Z}_{m})\nonumber\\
&\quad+\frac{1}{N}\underset{m\in \mathscr{I}}{\overset{}{\sum}}(Y^{i;F}_{m}-\hat{g}^{i}(\mathbf{Z}_{m}))\hat{A}_{m}^{i}+\hat{\kappa}_{N}^{i;CF},\tag{\ref{eqt:higher-order orthogonal condition before estimator in appendix}}\\
\hat{\bar{\theta}}_{N}^{i}&=\frac{1}{N}\underset{m=1}{\overset{N}{\sum}}\hat{g}^{i}(\mathbf{Z}_{m})\nonumber\\
&\quad+\frac{1}{N}\underset{m\in \mathscr{I}}{\overset{}{\sum}}(Y^{i;F}_{m}-\hat{g}^{i}(\mathbf{Z}_{m}))\hat{A}_{m}^{i}+\hat{\kappa}_{N}^{i;F}\tag{\ref{eqt:higher-order orthogonal condition final estimator in appendix}}
\end{align}
}\noindent
for simplicity. In addition, we have to use two lemmas and two propositions to study the consistency of $\hat{\theta}_{N}^{i}$. We state them with the proofs.
\begin{lemma}\label{lemma:simple lemma}	
Given two sequences of random variables $(X_{N})_{N=1}^{\infty}$ and $(Y_{N})_{N=1}^{\infty}$ such that $X_{N}\overset{d}{=}Y_{N}$. If $X_{N}\overset{p}{\rightarrow}c$ for some constant $c$, then $Y_{N}\overset{p}{\rightarrow}c$.	
\end{lemma}
\begin{proof}
Let $f_{X_{N}}(\cdot)$ and $f_{Y_{N}}(\cdot)$ be the density functions of the random variables $X_{N}$ and $Y_{N}$ respectively. Since $X_{N}\overset{d}{=}Y_{N}$, $f_{X_{N}}(\cdot)=f_{Y_{N}}(\cdot)$. Hence, $\mathbb{P}\left\{\left|X_{N}-c\right|\geq \epsilon\right\}=\int_{|z-c|\geq\epsilon}f_{X_{N}}(z)dz=\int_{|z-c|\geq\epsilon}f_{Y_{N}}(z)dz=\mathbb{P}\left\{\left|Y_{N}-c\right|\geq \epsilon\right\}$. Consequently, $X_{N}\overset{p}{\rightarrow}c$ implies $Y_{N}\overset{p}{\rightarrow}c$.
\end{proof}
\begin{lemma}\label{lemma:simple lemma2}	
Given random variables $X$, $Y$, $E$, $Z$. If $(X \overset{d}{=} Y) \mid Z$, $(X \perp \!\!\! \perp E) \mid Z$, $(Y \perp \!\!\! \perp E) \mid Z$, then $Xh(E,Z) \overset{d}{=} Yh(E,Z)$ for any function $h$.
\end{lemma}
\begin{proof}
Define $f_Z(z)$ as the density function of $Z$, $f_{X|Z}(x|z)$ is the conditional density function of $X|Z$, $f_{Y|Z}(y|z)$ is the conditional density function of $Y|Z$, $f_{E|Z}(e|z)$ is the conditional density function of $E|Z$, $f_{X,E|Z}(x,e|z)$ is the conditional joint density function of $X,E|Z$, and $f_{Y,E|Z}(y,e|z)$ is the conditional joint density function of $Y,E|Z$. For a measurable set $\mathcal{A}$, we have
\begin{equation*}
{\small
\begin{aligned}
&\mathbb{P}\{Xh(E,Z)\in \mathcal{A}\}\\
=&\int_{\Omega_{Z}}\left\{\iint_{\Omega_{X} \times \Omega_{E}} \mathbf{1}_{\{xh(e,z)\in \mathcal{A}\}}f_{X,E|Z}(x,e|z)dxde\right\}f_{Z}(z)dz\\
\overset{*}{=}&\int_{\Omega_{Z}}\left\{\iint_{\Omega_{X} \times \Omega_{E}} \mathbf{1}_{\{xh(e,z)\in \mathcal{A}\}}f_{X|Z}(x|z)f_{E|Z}(e|z)dxde\right\}f_{Z}(z)dz\\
\overset{\triangle}{=}&\int_{\Omega_{Z}}\left\{\iint_{\Omega_{Y} \times \Omega_{E}} \mathbf{1}_{\{yh(e,z)\in \mathcal{A}\}}f_{Y|Z}(y|z)f_{E|Z}(e|z)dyde\right\}f_{Z}(z)dz\\
\overset{\square}{=}&\int_{\Omega_{Z}}\left\{\iint_{\Omega_{Y} \times \Omega_{E}} \mathbf{1}_{\{yh(e,z)\in \mathcal{A}\}}f_{Y,E|Z}(y,e|z)dyde\right\}f_{Z}(z)dz\\
=&\mathbb{P}\{Yh(E,Z)\in \mathcal{A}\}
\end{aligned}	
}
\end{equation*}
$*$ holds since $(X \perp \!\!\! \perp E) \mid Z$, $\triangle$ holds since $(X \overset{d}{=} Y) \mid Z$, and $\square$ holds since $(Y \perp \!\!\! \perp E) \mid Z$.
\end{proof}
\begin{proposition}\label{lemma:difference from the statistical standpoint}
Suppose {\small $\mathbb{E}\left[(\xi^{i;F}_{m})^{2}\mid\mathbf{Z}\right]$} and {\small $(A^{i}_m)^2$} exist such that $\mathbb{E}\left[(A_m^{i})^2\mathbb{E}\left[(\xi_m^{i;F})^{2}\mid\mathbf{Z}\right]\right]$ is finite for all $m$. We have {\small $\kappa_{N}^{i;CF}-\kappa_{N}^{i;F}\overset{p}{\rightarrow} 0$} when {\small $N\rightarrow \infty$}.
\end{proposition}
\begin{proof}
$\forall \epsilon>0$, we consider $\mathbb{P}\left\{\left|\kappa_{N}^{i;CF}-\kappa_{N}^{i;F}\right|\geq\epsilon\right\}$. Indeed, we have
\begin{equation*}
{\small
\begin{aligned}
&\mathbb{P}\left\{\left|\kappa_{N}^{i;CF}-\kappa_{N}^{i;F}\right|\geq\epsilon\right\}\leq \frac{\mathbb{E}\left[\left(\kappa_{N}^{i;CF}-\kappa_{N}^{i;F}\right)^{2}\right]}{\epsilon^{2}}\\
=&\frac{\frac{1}{N^{2}}\mathbb{E}\left[\left(\underset{m\in \mathscr{I}^{c}}{\sum}(\xi_{m}^{i;CF}-\xi_{m}^{i;F})A_{m}^{i}\right)^{2}\right]}{\epsilon^{2}}.\\
\end{aligned}
}
\end{equation*}
Denoting $\xi_{m}^{i;CF}-\xi_{m}^{i;F}$ as $\Xi_{m}^{i}$, we have
\begin{equation*}
{\small
\begin{aligned}
&\mathbb{E}\big[\big(\underset{m\in \mathscr{I}^{c}}{\sum}(\xi_{m}^{i;CF}-\xi_{m}^{i;F})A_{m}^{i}\big)^{2}\big]=\mathbb{E}\big[\big(\underset{m\in \mathscr{I}^{c}}{\sum}\Xi_{m}^{i}A_{m}^{i}\big)^{2}\big]\\
=&\mathbb{E}\big[\underset{m,\bar{m}\in \mathscr{I}^{c}}{\sum}\Xi_{m}^{i}A_{m}^{i}\Xi_{\bar{m}}^{i}A_{\bar{m}}^{i}\big]=\underset{m,\bar{m}\in \mathscr{I}^{c}}{\sum}\mathbb{E}\left[\Xi_{m}^{i}A_{m}^{i}\Xi_{\bar{m}}^{i}A_{\bar{m}}^{i}\right]\\
&\;+\underset{\substack{m,\bar{m}\in \mathscr{I}^{c} \\ m\neq \bar{m}}}{\sum}\mathbb{E}\left[A_{m}^{i}A_{\bar{m}}^{i}\mathbb{E}\left[\Xi_{m}^{i}\Xi_{\bar{m}}^{i}\mid D,\mathbf{Z}\right]\right]\\
=&\underset{m\in \mathscr{I}^{c}}{\sum}\mathbb{E}\left[(A_{m}^{i})^{2}\mathbb{E}\left[(\Xi_{m}^{i})^{2}\mid D,\mathbf{Z}\right]\right]\\
&\;+\underset{\substack{m,\bar{m}\in \mathscr{I}^{c} \\ m\neq \bar{m}}}{\sum}\mathbb{E}\left[A_{m}^{i}A_{\bar{m}}^{i}\mathbb{E}\left[\Xi_{m}^{i}\mid D,\mathbf{Z}\right]\mathbb{E}\left[\Xi_{\bar{m}}^{i}\mid D,\mathbf{Z}\right]\right]\\
=&\underset{m\in \mathscr{I}^{c}}{\sum}\mathbb{E}\left[(A_{m}^{i})^{2}\mathbb{E}\left[(\Xi_{m}^{i})^{2}\mid \mathbf{Z}\right]\right]=2\underset{m\in \mathscr{I}^{c}}{\sum}\mathbb{E}\left[(A_{m}^{i})^{2}\mathbb{E}\left[(\xi^{i;F}_{m})^{2}\mid \mathbf{Z}\right]\right]\\
\leq& 2N\mathbb{E}\left[(A^{i})^{2}\mathbb{E}\left[(\xi^{i;F})^{2}\mid \mathbf{Z}\right]\right].
\end{aligned}
}
\end{equation*}
The last equality follows from
\begin{equation*}
{\small
\begin{aligned}
&\mathbb{E}\big[(\Xi^{i}_{m})^{2}\mid\mathbf{Z}\big]=\mathbb{E}\big[(\xi^{i;CF}_{m}-\xi^{i;F}_{m})^{2}\mid\mathbf{Z}\big]\\
=&\mathbb{E}\big[(\xi^{i;CF}_{m})^{2}\mid\mathbf{Z}\big]-2\mathbb{E}\big[\xi^{i;F}_{m}\xi^{i;CF}_{m}\mid\mathbf{Z}\big]+\mathbb{E}\big[(\xi^{i;F}_{m})^{2}\mid\mathbf{Z}\big]\\
=&\mathbb{E}\big[(\xi^{i;CF}_{m})^{2}\mid\mathbf{Z}\big]-2\mathbb{E}\big[\xi^{i;F}_{m}\mid\mathbf{Z}\big]\mathbb{E}\big[\xi^{i;CF}_{m}\mid\mathbf{Z}\big]+\mathbb{E}\big[(\xi^{i;F}_{m})^{2}\mid\mathbf{Z}\big]\\
=&\mathbb{E}\big[(\xi^{i;CF}_{m})^{2}\mid\mathbf{Z}\big]+\mathbb{E}\big[(\xi^{i;F}_{m})^{2}\mid\mathbf{Z}\big]=2\mathbb{E}\big[(\xi^{i;F}_{m})^{2}\mid\mathbf{Z}\big].
\end{aligned}
}
\end{equation*}
As a consequence, we have
\begin{equation*}
{\small
\begin{aligned}
\mathbb{P}\left\{\left|\kappa_{N}^{i;CF}-\kappa_{N}^{i;F}\right|\geq\epsilon\right\}&\leq \frac{2N\mathbb{E}\left[(A^{i})^{2}\mathbb{E}\left[(\xi^{i;F})^{2}\mid\mathbf{Z}\right]\right]}{N^{2}\epsilon^{2}}\\
&=\frac{2\mathbb{E}\left[(A^{i})^{2}\mathbb{E}\left[(\xi^{i;F})^{2}\mid\mathbf{Z}\right]\right]}{N\epsilon^{2}}\rightarrow 0
\end{aligned}
}
\end{equation*}
when $N\rightarrow \infty$. As a result, we have {\small $\kappa_{N}^{i;CF}-\kappa_{N}^{i;F}\overset{p}{\rightarrow} 0$}. The proof is completed.
\end{proof}

\begin{proposition}\label{lemma:unbiasedness and consistency}
Suppose that, conditioning on {\small $\mathbf{Z}$}, {\small $\xi_{m,u}^{i;F}$} are i.i.d. of {\small $\xi_{m}^{i;F}$} and {\small $\xi_{m,u}^{i;F}$} are i.i.d. of {\small $\xi_{m,\bar{u}}^{i;F}$} {\small $\forall u,\;\bar{u}\in\{1,2,\cdots,R\}$}. We have
\begin{equation*}
\kappa_{N}^{i;F}-\kappa_{R,N}^{i;F}\overset{p}{\rightarrow}0\quad \textit{when}\quad N\rightarrow \infty.\label{item:statistical difference 2}
\end{equation*}
\end{proposition}
\begin{proof}\ \\
Write
\begin{equation*}
{\small
\begin{aligned}
\kappa_{R,N}^{i;F}&=\frac{1}{R}\overset{R}{\underset{u=1}{\sum}}\bigg[\frac{1}{N}\underset{m\in \mathscr{I}^{c}}{\overset{}{\sum}}\xi_{m,u}^{i;F} A_{m}^{i}\bigg]\\
&=\frac{1}{N}\underset{m\in \mathscr{I}^{c}}{\overset{}{\sum}}\left(\frac{1}{R}\overset{R}{\underset{u=1}{\sum}}\xi_{m,u}^{i;F}\right) A_{m}^{i}=\frac{1}{N}\underset{m\in \mathscr{I}^{c}}{\overset{}{\sum}}\mathscr{E}_{m}^{i} A_{m}^{i}.
\end{aligned}
}
\end{equation*}
$\forall\epsilon>0$, we have
\begin{equation*}
{\small
\begin{aligned}
\mathbb{P}\left\{\left|\kappa_{N}^{i;F}-\kappa_{R,N}^{i;F}\right|\geq\epsilon\right\}\leq  \frac{\mathbb{E}\left[\left(\frac{1}{N}\underset{m\in \mathscr{I}^{c}}{\overset{}{\sum}}\left[\mathscr{E}_{m}^{i}-\xi_{m}^{i;F}\right]A_{m}^{i}\right)^{2}\right]}{\epsilon^{2}}.
\end{aligned}
}
\end{equation*}
We simplify {\small $\mathbb{E}\left[\left(\frac{1}{N}\underset{m\in \mathscr{I}^{c}}{\overset{}{\sum}}\left[\mathscr{E}_{m}^{i}-\xi_{m}^{i;F}\right]A_{m}^{i}\right)^{2}\right]$}. Note that
\begin{equation*}
\resizebox{0.5\textwidth}{!}{$
\begin{aligned}
&\mathbb{E}\left[\left(\frac{1}{N}\underset{m\in \mathscr{I}^{c}}{\overset{}{\sum}}\left[\mathscr{E}_{m}^{i}-\xi_{m}^{i;F}\right]A_{m}^{i}\right)^{2}\right]\\
=&\frac{1}{N^{2}}\:\:\sum_{\mathclap{\substack{m,\bar{m}\in \mathscr{I}^{c}}}}\mathbb{E}\left[\left(\mathscr{E}_{m}^{i}-\xi_{m}^{i;F}\right)\left(\mathscr{E}_{\bar{m}}^{i}-\xi_{\bar{m}}^{i;F}\right)A_{\bar{m}}^{i}A_{m}^{i}\right]\\
=&\frac{1}{N^{2}}\:\:\sum_{\mathclap{\substack{m\in \mathscr{I}^{c}}}}\mathbb{E}\left[\left(\mathscr{E}_{m}^{i}-\xi_{m}^{i;F}\right)^{2}(A_{m}^{i})^{2}\right]\\
&\;+\frac{1}{N^{2}}\:\:\sum_{\mathclap{\substack{m,\bar{m}\in \mathscr{I}^{c}\\ m\neq\bar{m}}}}\mathbb{E}\left[\left(\mathscr{E}_{m}^{i}-\xi_{m}^{i;F}\right)\left(\mathscr{E}_{\bar{m}}^{i}-\xi_{\bar{m}}^{i;F}\right)A_{\bar{m}}^{i}A_{m}^{i}\right]\\
=&\frac{1}{N^{2}}\:\:\sum_{\mathclap{\substack{m\in \mathscr{I}^{c}}}}\mathbb{E}\left[\left(\mathscr{E}_{m}^{i}-\xi_{m}^{i;F}\right)^{2}(A_{m}^{i})^{2}\right]\\
&\;+\frac{1}{N^{2}}\:\:\sum_{\mathclap{\substack{m,\bar{m}\in \mathscr{I}^{c}\\ m\neq\bar{m}}}}\mathbb{E}\left[A_{\bar{m}}^{i}A_{m}^{i}\mathbb{E}\left[\left(\mathscr{E}_{m}^{i}-\xi_{m}^{i;F}\right)\left(\mathscr{E}_{\bar{m}}^{i}-\xi_{\bar{m}}^{i;F}\right)\mid D,\mathbf{Z}\right]\right]\\
=&\frac{1}{N^{2}}\:\:\sum_{\mathclap{\substack{m\in \mathscr{I}^{c}}}}\mathbb{E}\left[\left(\mathscr{E}_{m}^{i}-\xi_{m}^{i;F}\right)^{2}(A_{m}^{i})^{2}\right]\\
&\;+\frac{1}{N^{2}}\:\:\sum_{\mathclap{\substack{m,\bar{m}\in \mathscr{I}^{c}\\ m\neq\bar{m}}}}\mathbb{E}\left[A_{\bar{m}}^{i}A_{m}^{i}\mathbb{E}\left[\left(\mathscr{E}_{m}^{i}-\xi_{m}^{i;F}\right)\mid\mathbf{Z}\right]\mathbb{E}\left[\left(\mathscr{E}_{\bar{m}}^{i}-\xi_{\bar{m}}^{i;F}\right)\mid\mathbf{Z}\right]\right]\\
=&\frac{1}{N^{2}}\:\:\sum_{\mathclap{\substack{m\in \mathscr{I}^{c}}}}\mathbb{E}\left[\left(\mathscr{E}_{m}^{i}-\xi_{m}^{i;F}\right)^{2}(A_{m}^{i})^{2}\right].
\end{aligned}$
}
\end{equation*}
The last equality in the above derivation follows from the fact that, conditioning on {\small $\mathbf{Z}$}, {\small $\xi_{m,u}^{i;F}$} are i.i.d. of {\small $\xi_{m}^{i;F}$} for any {\small $u\in\{1,2,\cdots,R\}$}. Indeed, we have $\mathbb{E}\left[\xi_{m,u}^{i;F}\mid\mathbf{Z}\right]=\mathbb{E}\left[\xi_{u}^{i;F}\mid\mathbf{Z}\right]$ for any $m$ and $u\in\{1,2,\cdots,R\}$. Consequently, we have
\begin{equation*}
\resizebox{0.49\textwidth}{!}{$\begin{aligned}
&\mathbb{E}\left[\left(\mathscr{E}_{m}^{i}-\xi_{m}^{i;F}\right)\mid\mathbf{Z}\right]=\mathbb{E}\left[\left(\frac{1}{R}\overset{R}{\underset{u=1}{\sum}} \xi_{m,u}^{i;F}-\xi_{m}^{i;F}\right)\mid\mathbf{Z}\right]\\
=&\frac{1}{R}\overset{R}{\underset{u=1}{\sum}} \mathbb{E}\left[\xi_{m,u}^{i;F}\mid\mathbf{Z}\right]-\mathbb{E}\left[\xi_{m}^{i;F}\mid\mathbf{Z}\right]=\frac{1}{R}\overset{R}{\underset{u=1}{\sum}} \mathbb{E}\left[\xi_{m}^{i;F}\mid\mathbf{Z}\right]-\mathbb{E}\left[\xi_{m}^{i;F}\mid\mathbf{Z}\right]\\
=& \mathbb{E}\left[\xi_{m}^{i;F}\mid\mathbf{Z}\right]-\mathbb{E}\left[\xi_{m}^{i;F}\mid\mathbf{Z}\right]=0.
\end{aligned}
$}
\end{equation*}
In addition, we simplify the quantity {\small $\mathbb{E}\left[\left(\mathscr{E}_{m}^{i}-\xi_{m}^{i;F}\right)^{2}\right]$}. Note that {\small $\mathscr{E}_{m}^{i}-\xi_{m}^{i;F}=\frac{1}{R}\underset{u=1}{\overset{R}{\sum}}[\xi_{m,u}^{i;F}-\xi_{m}^{i;F}]$}. We therefore have
\begin{equation*}
{\small
\begin{aligned}
&\mathbb{E}\left[\left(\overset{R}{\underset{u=1}{\sum}}\left[\xi_{m,u}^{i;F}-\xi_{m}^{i;F}\right]\right)^{2}\right]\\
=&\overset{R}{\underset{u,\bar{u}=1}{\sum}}\mathbb{E}\left[\left(\xi_{m,u}^{i;F}-\xi_{m}^{i;F}\right)\left(\xi_{m,\bar{u}}^{i;F}-\xi_{m}^{i;F}\right)\right]\\
=&\overset{R}{\underset{u,\bar{u}=1}{\sum}}\left\{\mathbb{E}\left[\xi_{m,u}^{i;F}\xi_{m,\bar{u}}^{i;F}\right]-\mathbb{E}\left[\xi_{m,u}^{i;F}\xi_{m}^{i;F}\right]\right.\\
&\quad\quad\quad\left.-\mathbb{E}\left[\xi_{m}^{i;F}\xi_{m,\bar{u}}^{i;F}\right]+\mathbb{E}\left[\xi_{m}^{i;F}\xi_{m}^{i;F}\right]\right\}\\
=&\overset{R}{\underset{u=1}{\sum}}\mathbb{E}\left[(\xi_{m,u}^{i;F})^{2}\right]-2R\overset{R}{\underset{u=1}{\sum}}\mathbb{E}\left[\xi_{m,u}^{i;F}\xi_{m}^{i;F}\right]+R^{2}\mathbb{E}\left[(\xi_{m}^{i;F})^{2}\right]\\
&\;+\overset{R}{\underset{\substack{u,\bar{u}=1 \\ u\neq\bar{u}}}{\sum}}\mathbb{E}\left[\xi_{m,u}^{i;F}\xi_{m,\bar{u}}^{i;F}\right]=\left[R^{2}+R\right]\mathbb{E}\left[\left(\xi_{m}^{i;F}\right)^{2}\right].
\end{aligned}
}
\end{equation*}
We justify the last equality. The last equality follows from the fact that, conditioning on {\small $\mathbf{Z}$}, {\small $\xi_{m,u}^{i;F}$} are i.i.d. of {\small $\xi_{m}^{i;F}$} and {\small $\xi_{m,u}^{i;F}$} are i.i.d. of {\small $\xi_{m,\bar{u}}^{i;F}$} for any {\small $u,\;\bar{u}\in\{1,2,\cdots,R\}$}. Indeed, under the given fact, we have
\begin{equation*}
\resizebox{0.5\textwidth}{!}{$
\begin{aligned}
&\mathbb{E}\left[\xi_{m,u}^{i;F}\xi_{m}^{i;F}\right]=\mathbb{E}\left[\mathbb{E}\left[\xi_{m,u}^{i;F}\xi_{m}^{i;F}\mid \mathbf{Z}\right]\right]=\mathbb{E}\left[\mathbb{E}\left[\xi_{m,u}^{i;F}\mid \mathbf{Z}\right]\mathbb{E}\left[\xi_{m}^{i;F}\mid \mathbf{Z}\right]\right]\\
=&\mathbb{E}\left[\mathbb{E}\left[\mathbb{E}\left[\xi_{m,u}^{i;F}\mid D, \mathbf{Z}\right]\mid\mathbf{Z}\right]\mathbb{E}\left[\mathbb{E}\left[\xi_{m}^{i;F}\mid D,\mathbf{Z}\right]\mid \mathbf{Z}\right]\right]=0
\end{aligned}
$}
\end{equation*}
and
\begin{equation*}
\resizebox{0.5\textwidth}{!}{$
\begin{aligned}
&\mathbb{E}\left[\xi_{m,u}^{i;F}\xi_{m,\bar{u}}^{i;F}\right]=\mathbb{E}\left[\mathbb{E}\left[\xi_{m,u}^{i;F}\xi_{m,\bar{u}}^{i;F}\mid \mathbf{Z}\right]\right]=\mathbb{E}\left[\mathbb{E}\left[\xi_{m,u}^{i;F}\mid \mathbf{Z}\right]\mathbb{E}\left[\xi_{m,\bar{u}}^{i;F}\mid \mathbf{Z}\right]\right]\\
=&\mathbb{E}\left[\mathbb{E}\left[\mathbb{E}\left[\xi_{m,u}^{i;F}\mid D,\mathbf{Z}\right]\mid \mathbf{Z}\right]\mathbb{E}\left[\mathbb{E}\left[\xi_{m,\bar{u}}^{i;F}\mid D,\mathbf{Z}\right]\mid \mathbf{Z}\right]\right]=0.
\end{aligned}
$}
\end{equation*}
Consequently, we have
\begin{equation*}
{\small
\begin{aligned}
\mathbb{E}\left[\left(\mathscr{E}_{m}^{i}-\xi_{m}^{i;F}\right)^{2}\right]=\left(1+\frac{1}{R}\right)\mathbb{E}\left[\left(\xi_{m}^{i;F}\right)^{2}\right].
\end{aligned}
}
\end{equation*}
Thus, we have
\begin{equation}
{\small
\begin{aligned}\label{eqt:summation R bound}
\mathbb{P}\left\{\left|\kappa_{N}^{i;F}-\kappa_{R,N}^{i;F}\right|\geq\epsilon\right\}&\leq \frac{\frac{1}{N^{2}}\underset{m\in \mathscr{I}^{c}}{\overset{}{\sum}}\left(1+\frac{1}{R}\right)\mathbb{E}\left[\left(\xi_{m}^{i;F}\right)^{2}\right]}{\epsilon^{2}}\\
&\leq \frac{\left(1+\frac{1}{R}\right)\mathbb{E}\left[\left(\xi^{i;F}\right)^{2}\right]}{N\epsilon^{2}}.
\end{aligned}
}
\end{equation}
We notice that no matter we set $R\rightarrow \infty$ followed by $N\rightarrow \infty$ or vice versa, or we fix $R$ but let $N\rightarrow \infty$, we see that {\small $\mathbb{P}\left\{\left|\kappa_{N}^{i;F}-\kappa_{R,N}^{i;F}\right|\geq\epsilon\right\} \rightarrow 0$}.
\end{proof}
Now, we are ready to investigate if the estimator $\hat{\theta}_{N}^{i}$ is a consistent estimator of $\theta^{i}$. Our goal is to show that $\mathbb{P}_{\hat{\rho}}\left\{\left|\hat{\theta}_{N}^{i}-\theta^{i}\right| \geq \epsilon\right\}\overset{p}{\rightarrow} 0$. Before presenting the proof, we notice that $\left(\hat{\xi}^{i;F}\overset{d}{=}\hat{\xi}^{i;CF}\right)\mid\mathbf{Z}$  when $\hat{g}^{i}$ and $g^{i}$ satisfie Assumption \ref{RCL assumption}. This is because that $\hat{\xi}^{i;F}=Y^{i;F}-\hat{g}^{i}(\mathbf{Z})=\xi^{i;F} + g^{i}(\mathbf{Z})-\hat{g}^{i}(\mathbf{Z})$, $\hat{\xi}^{i;CF}=Y^{i;CF}-\hat{g}^{i}(\mathbf{Z})=\xi^{i;CF} + g^{i}(\mathbf{Z})-\hat{g}^{i}(\mathbf{Z})$, {\small $(\xi^{i;F} \perp \!\!\! \perp D) \mid \mathbf{Z}$}, {\small $(\xi^{i;CF} \perp \!\!\! \perp D) \mid \mathbf{Z}$}, and $\left(\xi^{i;F}\overset{d}{=}\xi^{i;CF}\right)\mid\mathbf{Z}$.
\begin{proof}
$\forall\epsilon>0$, we have
\begin{equation*}
{\small
\begin{aligned}
&\mathbb{P}_{\hat{\rho}}\left\{\left|\hat{\theta}_{N}^{i}-\theta^{i}\right| \geq \epsilon\right\}\\
=&\mathbb{P}_{\hat{\rho}}\left\{\left|\hat{\theta}_{N}^{i}-\hat{\bar{\theta}}_{N}^{i}+\hat{\bar{\theta}}_{N}^{i}-\theta^{i}\right| \geq \epsilon\right\}\\
\leq&\mathbb{P}_{\hat{\rho}}\left\{\left|\hat{\theta}_{N}^{i}-\hat{\bar{\theta}}_{N}^{i}\right| \geq \frac{\epsilon}{2}\right\}+\mathbb{P}_{\hat{\rho}}\left\{\left|\hat{\bar{\theta}}_{N}^{i}-\theta^{i}\right| \geq \frac{\epsilon}{2}\right\}.
\end{aligned}
}
\end{equation*}
Since {\small $\left(\xi^{i;F}\overset{d}{=}\xi^{i;CF}\right)\mid\mathbf{Z}$} and {\small $\left(\hat{\xi}^{i;F}\overset{d}{=}\hat{\xi}^{i;CF}\right)\mid\mathbf{Z}$}, we have {\small $\hat{\bar{\theta}}_{N}^{i}\overset{d}{=}\hat{\tilde{\theta}}_{N}^{i}$} by Lemma \ref{lemma:simple lemma2}. Moreover, we know that {\small $\mathbb{P}_{\hat{\rho}}\left\{\left|\hat{\tilde{\theta}}_{N}^{i}-\theta^{i}\right| \geq \frac{\epsilon}{2}\right\}\overset{p}{\rightarrow}0$} under the assumptions given in \cite{mackey2018orthogonal}. Together with the fact that {\small $\hat{\bar{\theta}}_{N}^{i}\overset{d}{=}\hat{\tilde{\theta}}_{N}^{i}$}, we have {\small $\mathbb{P}_{\hat{\rho}}\left\{\left|\hat{\bar{\theta}}_{N}^{i}-\theta^{i}\right| \geq \frac{\epsilon}{2}\right\}\overset{p}{\rightarrow}0$} by Lemma \ref{lemma:simple lemma}. We turn to consider the quantity {\small $\mathbb{P}_{\hat{\rho}}\left\{\left|\hat{\theta}_{N}^{i}-\hat{\bar{\theta}}_{N}^{i}\right| \geq \frac{\epsilon}{2}\right\}$}, and we aim to show that {\small $\mathbb{P}_{\hat{\rho}}\left\{\left|\hat{\theta}_{N}^{i}-\hat{\bar{\theta}}_{N}^{i}\right| \geq \frac{\epsilon}{2}\right\}\overset{p}{\rightarrow} 0$}. Notice that
\begin{equation}
{\small
\begin{aligned}\label{eqt:consistent proof estimator decomp}
&\mathbb{P}_{\hat{\rho}}\left\{\left|\hat{\theta}_{N}^{i}-\hat{\bar{\theta}}_{N}^{i}\right| \geq \frac{\epsilon}{2}\right\}=\mathbb{P}_{\hat{\rho}}\left\{\left|\hat{\kappa}_{R,N}^{i;F}-\hat{\kappa}_{N}^{i;F}\right| \geq \frac{\epsilon}{2}\right\}\\
\leq&\underbrace{\mathbb{P}_{\hat{\rho}}\left\{\left|\hat{\kappa}_{R,N}^{i;F}-\kappa_{R,N}^{i;F}\right| \geq \frac{\epsilon}{8}\right\}}_{(a)}+\underbrace{\mathbb{P}_{\hat{\rho}}\left\{\left|\kappa_{R,N}^{i;F}-\kappa_{N}^{i;F}\right| \geq \frac{\epsilon}{8}\right\}}_{(b)}\\
&\;+\underbrace{\mathbb{P}_{\hat{\rho}}\left\{\left|\kappa_{N}^{i;F}-\kappa_{N}^{i;CF}\right| \geq \frac{\epsilon}{8}\right\}}_{(c)}+\underbrace{\mathbb{P}_{\hat{\rho}}\left\{\left|\kappa_{N}^{i;CF}-\hat{\kappa}_{N}^{i;F}\right| \geq \frac{\epsilon}{8}\right\}}_{(d)}.
\end{aligned}
}
\end{equation}
Note that {\small $\left|\kappa_{R,N}^{i;F}-\kappa_{N}^{i;F}\right|$} and {\small $\left|\kappa_{N}^{i;F}-\kappa_{N}^{i;CF}\right|$} do not incorporate any terms related to the estimated function {\small $\hat{\rho}$}. From Proposition \ref{lemma:difference from the statistical standpoint} and Proposition \ref{lemma:unbiasedness and consistency}, we conclude that (\ref{eqt:consistent proof estimator decomp}b) and (\ref{eqt:consistent proof estimator decomp}c) converge to {\small $0$} in probability respectively. 
It remains to show the convergence of (\ref{eqt:consistent proof estimator decomp}a) and (\ref{eqt:consistent proof estimator decomp}d).
Consider (\ref{eqt:consistent proof estimator decomp}d) first. Since
\begin{equation*}
{\small
\begin{aligned}
&\left|\kappa_{N}^{i;CF}-\hat{\kappa}_{N}^{i;F}\right|=\left|\frac{1}{N}\underset{m\in \mathscr{I}^{c}}{\sum}\xi_{m}^{i;CF}A_{m}^{i}-\frac{1}{N}\underset{m\in \mathscr{I}^{c}}{\sum}\hat{\xi}_{m}^{i;F}\hat{A}_{m}^{i}\right|\\
\leq&\underbrace{\left|\frac{1}{N}\underset{m\in \mathscr{I}^{c}}{\sum}(\xi_{m}^{i;CF}A_{m}^{i}-\xi_{m}^{i;F}A_{m}^{i})\right|}_{\Gamma_{1}}\\
&\;+\underbrace{\left|\frac{1}{N}\underset{m\in \mathscr{I}^{c}}{\sum}(\xi_{m}^{i;F}A_{m}^{i}-\hat{\xi}_{m}^{i;F}A_{m}^{i})\right|}_{\Gamma_{2}}+\underbrace{\left|\frac{1}{N}\underset{m\in \mathscr{I}^{c}}{\sum}\hat{\xi}_{m}^{i;F}(A_{m}^{i}-\hat{A}_{m}^{i})\right|}_{\Gamma_{3}},
\end{aligned}
}
\end{equation*}
(\ref{eqt:consistent proof estimator decomp}d) is bounded above by
\begin{equation}
{\small
\begin{aligned}
&\underbrace{\mathbb{P}_{\hat{\rho}}\left\{\Gamma_{1} \geq \frac{\epsilon}{24}\right\}}_{(a)}+\underbrace{\mathbb{P}_{\hat{\rho}}\left\{\Gamma_{2} \geq \frac{\epsilon}{24}\right\}}_{(b)}+\underbrace{\mathbb{P}_{\hat{\rho}}\left\{\Gamma_{3} \geq \frac{\epsilon}{24}\right\}}_{(c)}.\label{eqt:consistency check equation}
\end{aligned}
}\\
\end{equation}
(\ref{eqt:consistency check equation}a) converges to $0$ in probability due to Proposition \ref{lemma:difference from the statistical standpoint}. We study the quantities (\ref{eqt:consistency check equation}b) and (\ref{eqt:consistency check equation}c). 
	
(\ref{eqt:consistency check equation}b) can be further bounded. If {\small $N^{c}$} is the size of {\small $\mathscr{I}^{c}$}, then we have
\begin{equation*}
{\small
\begin{aligned}
\Gamma_{2}=&\left|\frac{1}{N}\underset{m\in \mathscr{I}^{c}}{\sum}(\xi_{m}^{i;F}A_{m}^{i}-\hat{\xi}_{m}^{i;F}A_{m}^{i})\right|\\
\leq&\underbrace{\left|\frac{1}{N}\underset{m\in \mathscr{I}^{c}}{\sum}\xi_{m}^{i;F}A_{m}^{i}-\frac{N^{c}}{N}\mathbb{E}_{\hat{\rho}}\left[\xi^{i;F}A^{i}\right]\right|}_{\Gamma_{2;1}}\\
&\;+\underbrace{\left|\frac{N^{c}}{N}\mathbb{E}_{\hat{\rho}}\left[\xi^{i;F}A^{i}\right]-\frac{N^{c}}{N}\mathbb{E}_{\hat{\rho}}\left[\hat{\xi}^{i;F}A^{i}\right]\right|}_{\Gamma_{2;2}}\\
&\;+\underbrace{\left|\frac{N^{c}}{N}\mathbb{E}_{\hat{\rho}}\left[\hat{\xi}^{i;F}A^{i}\right]-\frac{1}{N}\underset{m\in \mathscr{I}^{c}}{\sum}\hat{\xi}_{m}^{i;F}A_{m}^{i}\right|}_{\Gamma_{2;3}}.
\end{aligned}
}
\end{equation*}
We see that (\ref{eqt:consistency check equation}b) can be further bounded by
\begin{equation}
{\small
\begin{aligned}
&\underbrace{\mathbb{P}_{\hat{\rho}}\left\{\Gamma_{2;1} \geq \frac{\epsilon}{72}\right\}}_{(a)}+\underbrace{\mathbb{P}_{\hat{\rho}}\left\{\Gamma_{2;2} \geq \frac{\epsilon}{72}\right\}}_{(b)}+\underbrace{\mathbb{P}_{\hat{\rho}}\left\{\Gamma_{2;3} \geq \frac{\epsilon}{72}\right\}}_{(c)}.\label{eqt:term 1 check}
\end{aligned}
}
\end{equation}
We investigate if (\ref{eqt:term 1 check}a), (\ref{eqt:term 1 check}b), and (\ref{eqt:term 1 check}c) converge to $0$ in probability. We consider (\ref{eqt:term 1 check}a) first. Recall the assumptions that {\small $(\xi^{i;F} \perp \!\!\! \perp D) \mid \mathbf{Z}$}, {\small $\left(\xi^{i;F}\overset{d}{=}\xi^{i;CF}\right) \mid \mathbf{Z}$} and {\small $(\xi^{i;CF} \perp \!\!\! \perp D) \mid \mathbf{Z}$}, we have {\small $\frac{1}{N}\underset{m\in \mathscr{I}^{c}}{\sum}\xi_{m}^{i;F}A_{m}^{i}\overset{d}{=}\frac{1}{N}\underset{m\in \mathscr{I}^{c}}{\sum}\xi_{m}^{i;CF}A_{m}^{i}$} by Lemma \ref{lemma:simple lemma2}.	Since {\small $\Gamma_{2;1}=\frac{N^{c}}{N}\left|\frac{1}{N^{c}}\underset{m\in \mathscr{I}^{c}}{\sum}\xi_{m}^{i;F}A_{m}^{i}-\mathbb{E}_{\hat{\rho}}\left[\xi^{i;F}A^{i}\right]\right|$}, we have
\begin{equation*}
{\small
\begin{aligned}
&\mathbb{P}_{\hat{\rho}}\left\{\Gamma_{2;1}\geq\frac{\epsilon}{72}\right\}\\
=&\mathbb{P}_{\hat{\rho}}\left\{\left|\frac{1}{N^{c}}\underset{m\in \mathscr{I}^{c}}{\sum}\xi_{m}^{i;F}A_{m}^{i}-\mathbb{E}_{\hat{\rho}}\left[\xi^{i;F}A^{i}\right]\right|\geq\frac{\epsilon}{72}\cdot\frac{N}{N^{c}}\right\}\\
\leq& \mathbb{P}_{\hat{\rho}}\left\{\left|\frac{1}{N^{c}}\underset{m\in \mathscr{I}^{c}}{\sum}\xi_{m}^{i;F}A_{m}^{i}-\mathbb{E}_{\hat{\rho}}\left[\xi^{i;F}A^{i}\right]\right|\geq\frac{\epsilon}{72}\right\}\\
\leq&\frac{\mathbb{E}_{\hat{\rho}}\left[\left|\frac{1}{N^{c}}\underset{m\in \mathscr{I}^{c}}{\sum}\xi_{m}^{i;F}A_{m}^{i}-\mathbb{E}_{\hat{\rho}}\left[\xi^{i;F}A^{i}\right]\right|^{2}\right]}{\left(\frac{\epsilon}{72}\right)^{2}}.
\end{aligned}
}
\end{equation*}
Consider {\small $\mathbb{E}_{\hat{\rho}}\left[\left|\frac{1}{N^{c}}\underset{m\in \mathscr{I}^{c}}{\sum}\xi_{m}^{i;F}A_{m}^{i}-\mathbb{E}_{\hat{\rho}}\left[\xi^{i;F}A^{i}\right]\right|^{2}\right]$}. 

Note that it equals
\begin{equation*}
\resizebox{0.5\textwidth}{!}{$
\begin{aligned}
&\frac{1}{(N^{c})^{2}}\underset{m\in \mathscr{I}^{c}}{\sum}\mathbb{E}_{\hat{\rho}}\left[\left|\xi_{m}^{i;F}A_{m}^{i}-\mathbb{E}_{\hat{\rho}}\left[\xi^{i;F}A^{i}\right]\right|^{2}\right]\\
&\;+\frac{1}{\left(N^{c}\right)^{2}}\underset{\substack{m,\bar{m}\in \mathscr{I}^{c}\\m\neq\bar{m}}}{\sum}\mathbb{E}_{\hat{\rho}}\left[(\xi_{m}^{i;F}A_{m}^{i}-\mathbb{E}_{\hat{\rho}}[\xi^{i;F}A^{i}])(\xi_{\bar{m}}^{i;F}A_{\bar{m}}^{i}-\mathbb{E}_{\hat{\rho}}[\xi^{i;F}A^{i}])\right]\\
=&\frac{1}{\left(N^{c}\right)^{2}}\underset{m\in \mathscr{I}^{c}}{\sum}\mathbb{E}_{\hat{\rho}}\left[\left(A_{m}^{i}\right)^{2}\mathbb{E}_{\hat{\rho}}\left[\left(\xi_{m}^{i;F}\right)^{2}\mid D,\mathbf{Z}\right]\right]\\
&\;+\frac{1}{\left(N^{c}\right)^{2}}\underset{\substack{m,\bar{m}\in \mathscr{I}^{c}\\m\neq\bar{m}}}{\sum}\mathbb{E}_{\hat{\rho}}\left[A_{m}^{i}A_{\bar{m}}^{i}\mathbb{E}_{\hat{\rho}}\left[\xi_{m}^{i;F}\mid D,\mathbf{Z}\right]\mathbb{E}_{\hat{\rho}}\left[\xi_{\bar{m}}^{i;F}\mid D,\mathbf{Z}\right]\right]\\
=&\frac{1}{\left(N^{c}\right)^{2}}\underset{m\in \mathscr{I}^{c}}{\sum}\mathbb{E}_{\hat{\rho}}\left[\left(A_{m}^{i}\right)^{2}\mathbb{E}_{\hat{\rho}}\left[\left(\xi_{m}^{i;F}\right)^{2}\mid\mathbf{Z}\right]\right]\\
=&\frac{1}{N^{c}}\mathbb{E}_{\hat{\rho}}\left[\left(A^{i}\right)^{2}\mathbb{E}_{\hat{\rho}}\left[\left(\xi^{i;F}\right)^{2}\mid\mathbf{Z}\right]\right].
\end{aligned}
$}
\end{equation*}

Since {\small $A^{i}$} and {\small $\xi^{i;F}$} do not include the estimated nuisance parameters, {\small $\mathbb{E}_{\hat{\rho}}\left[\left(A^{i}\right)^{2}\mathbb{E}_{\hat{\rho}}\left[\left(\xi^{i;F}\right)^{2}\mid\mathbf{Z}\right]\right]$} is a constant. Moreover, note that {\small $N^{c}\rightarrow \infty$} when {\small $N\rightarrow \infty$}, we have
\begin{equation*}
{\small
\begin{aligned}
\mathbb{P}_{\hat{\rho}}\left\{\Gamma_{2;1}\geq\frac{\epsilon}{72}\right\}\leq\frac{72^2\;\mathbb{E}_{\hat{\rho}}\left[\left(A^{i}\right)^{2}\mathbb{E}_{\hat{\rho}}\left[\left(\xi^{i;F}\right)^{2}\mid\mathbf{Z}\right]\right]}{\epsilon^{2}N^{c}}\overset{p}{\rightarrow}0.
\end{aligned}
}
\end{equation*}
Now, we consider (\ref{eqt:term 1 check}b). Indeed, we have
\begin{equation*}
{\small
\begin{aligned}
&\mathbb{P}_{\hat{\rho}}\left\{\Gamma_{2;2}\geq\frac{\epsilon}{72}\right\}\\
=&\mathbb{P}_{\hat{\rho}}\left\{\left|\mathbb{E}_{\hat{\rho}}\left[\xi^{i;F}A^{i}\right]-\mathbb{E}_{\hat{\rho}}\left[\hat{\xi}^{i;F}A^{i}\right]\right|\geq\frac{\epsilon}{72}\cdot\frac{N}{N^{c}}\right\}\\
\leq&\mathbb{P}_{\hat{\rho}}\left\{\left|\mathbb{E}_{\hat{\rho}}\left[(\xi^{i;F}-\hat{\xi}^{i;F})A^{i}\right]\right|\geq\frac{\epsilon}{72}\right\}\\
\leq&\frac{72^2\left\{\mathbb{E}_{\hat{\rho}}\left[(\xi^{i;F}-\hat{\xi}^{i;F})A^{i}\right]\right\}^{2}}{\epsilon^{2}}\\
\leq&\frac{72^2\;\left\{\mathbb{E}_{\hat{\rho}}\left[(\xi^{i;F}-\hat{\xi}^{i;F})^{4q}\right]\right\}^{\frac{1}{2q}}\left\{\mathbb{E}_{\hat{\rho}}\left[(A^{i})^{\frac{4q}{4q-1}}\right]\right\}^{2-\frac{1}{2q}}}{\epsilon^{2}}\overset{p}{\rightarrow}0.
\end{aligned}
}
\end{equation*}
Here, the last inequality follows from the H\"{o}lders inequality, while the convergence holds $\forall q \in \{1,2,\dots,k\}$ according to  Assumption 1.5 of \cite{mackey2018orthogonal}. Finally, we consider (\ref{eqt:term 1 check}c). We can rewrite {\small $\Gamma_{2;3}$} as
\begin{equation*}
{\small
\begin{aligned}
\Gamma_{2;3}=\frac{N^{c}}{N}\left|\mathbb{E}_{\hat{\rho}}\left[\hat{\xi}^{i;F}A^{i}\right]-\frac{1}{N^{c}}\underset{m\in \mathscr{I}^{c}}{\sum}\hat{\xi}_{m}^{i;F}A_{m}^{i}\right|.
\end{aligned}
}
\end{equation*}
Now, we have
\begin{equation*}
{\small
\begin{aligned}
&\mathbb{P}_{\hat{\rho}}\left\{\Gamma_{2;3}\geq\frac{\epsilon}{72}\right\}\\
\leq&\mathbb{P}_{\hat{\rho}}\left\{\left|\mathbb{E}_{\hat{\rho}}\left[\hat{\xi}^{i;F}A^{i}\right]-\frac{1}{N^{c}}\underset{m\in \mathscr{I}^{c}}{\sum}\hat{\xi}_{m}^{i;F}A_{m}^{i}\right|\geq\frac{\epsilon}{72}\right\}\\
\leq&\frac{72^2\;\mathbb{E}_{\hat{\rho}}\left[\left\{\underset{m\in \mathscr{I}^{c}}{\sum}\left(\mathbb{E}_{\hat{\rho}}\left[\hat{\xi}^{i;F}A^{i}\right]-\hat{\xi}_{m}^{i;F}A_{m}^{i}\right)\right\}^{2}\right]}{\epsilon^{2}\left(N^{c}\right)^{2}}\\
=&\frac{72^2\;\underset{m\in \mathscr{I}^{c}}{\sum}\;\mathbb{E}_{\hat{\rho}}\left[\left(\mathbb{E}_{\hat{\rho}}\left[\hat{\xi}^{i;F}A^{i}\right]-\hat{\xi}_{m}^{i;F}A_{m}^{i}\right)^{2}\right]}{\epsilon^{2}\left(N^{c}\right)^{2}}\\
&\;+\frac{72^2\;\underset{\substack{m,\bar{m}\in \mathscr{I}^{c}\\m\neq\bar{m}}}{\sum}\mathbb{E}_{\hat{\rho}}[(\hat{\xi}_{m}^{i;F}-\xi_{m}^{i;F})A_{m}^{i}]\mathbb{E}_{\hat{\rho}}[(\hat{\xi}_{\bar{m}}^{i;F}-\xi_{\bar{m}}^{i;F})A_{\bar{m}}^{i}]}{\epsilon^{2}\left(N^{c}\right)^{2}}\\
&\;-2\frac{72^2\;\underset{\substack{m,\bar{m}\in \mathscr{I}^{c}\\m<\bar{m}}}{\sum}\mathbb{E}_{\hat{\rho}}[(\hat{\xi}_{m}^{i;F}-\xi_{m}^{i;F})A_{m}^{i}]\mathbb{E}_{\hat{\rho}}[(\hat{\xi}^{i;F}-\xi^{i;F})A^{i}]}{\epsilon^{2}\left(N^{c}\right)^{2}}\\
&\;+\frac{72^2\;\underset{\substack{m,\bar{m}\in \mathscr{I}^{c}\\m\neq\bar{m}}}{\sum}\left\{\mathbb{E}_{\hat{\rho}}\left[\left(\hat{\xi}^{i;F}-\xi^{i;F}\right)A^{i}\right]\right\}^{2}}{\epsilon^{2}\left(N^{c}\right)^{2}}.
\end{aligned}
}
\end{equation*}
Using Assumption 1.5 of \cite{mackey2018orthogonal}, we can conclude that {\small $\mathbb{P}_{\hat{\rho}}\left\{\Gamma_{2;3}\geq\frac{\epsilon}{72}\right\}\overset{p}{\rightarrow}0$.} Next, we come to bound (\ref{eqt:consistency check equation}c). Since {\small $N^{c}$} is the size of {\small $\mathscr{I}^{c}$} and
\begin{equation*}
{\small
\begin{aligned}
&\Gamma_{3}=\left|\frac{1}{N}\underset{m\in \mathscr{I}^{c}}{\sum}\hat{\xi}_{m}^{i;F}A_{m}^{i}-\frac{1}{N}\underset{m\in \mathscr{I}^{c}}{\sum}\hat{\xi}_{m}^{i;F}\hat{A}_{m}^{i}\right|\\
\leq&\left|\frac{1}{N}\underset{m\in \mathscr{I}^{c}}{\sum}\hat{\xi}_{m}^{i;F}A_{m}^{i}-\frac{N^{c}}{N}\mathbb{E}_{\hat{\rho}}\left[\hat{\xi}^{i;F}A^{i}\right]\right|\\
&\;+\left|\frac{N^{c}}{N}\mathbb{E}_{\hat{\rho}}\left[\hat{\xi}^{i;F}A^{i}\right]-\frac{N^{c}}{N}\mathbb{E}_{\hat{\rho}}\left[\hat{\xi}^{i;F}\hat{A}^{i}\right]\right|\\
&\;+\left|\frac{N^{c}}{N}\mathbb{E}_{\hat{\rho}}\left[\hat{\xi}^{i;F}\hat{A}^{i}\right]-\frac{1}{N}\underset{m\in \mathscr{I}^{c}}{\sum}\hat{\xi}_{m}^{i;F}\hat{A}_{m}^{i}\right|\\
=&\frac{N^{c}}{N}\underbrace{\left|\frac{1}{N^{c}}\underset{m\in \mathscr{I}^{c}}{\sum}\hat{\xi}_{m}^{i;F}A_{m}^{i}-\mathbb{E}_{\hat{\rho}}\left[\hat{\xi}^{i;F}A^{i}\right]\right|}_{\Gamma_{3;1}}\\
&\;+\frac{N^{c}}{N}\underbrace{\left|\mathbb{E}_{\hat{\rho}}\left[\hat{\xi}^{i;F}A^{i}\right]-\mathbb{E}_{\hat{\rho}}\left[\hat{\xi}^{i;F}\hat{A}^{i}\right]\right|}_{\Gamma_{3;2}}\\
&\;+\frac{N^{c}}{N}\underbrace{\left|\mathbb{E}_{\hat{\rho}}\left[\hat{\xi}^{i;F}\hat{A}^{i}\right]-\frac{1}{N^{c}}\underset{m\in \mathscr{I}^{c}}{\sum}\hat{\xi}_{m}^{i;F}\hat{A}_{m}^{i}\right|}_{\Gamma_{3;3}},
\end{aligned}
}
\end{equation*}
we see that (\ref{eqt:consistency check equation}c) can be further bounded by
\begin{equation}
{\small
\begin{aligned}
&\underbrace{\mathbb{P}_{\hat{\rho}}\left\{\Gamma_{3;1} \geq \frac{\epsilon}{72}\right\}}_{(a)}+\underbrace{\mathbb{P}_{\hat{\rho}}\left\{\Gamma_{3;2} \geq \frac{\epsilon}{72}\right\}}_{(b)}+\underbrace{\mathbb{P}_{\hat{\rho}}\left\{\Gamma_{3;3} \geq \frac{\epsilon}{72}\right\}}_{(c)}.\label{eqt:term 2 check}
\end{aligned}
}
\end{equation}
Similarly, we can prove that (\ref{eqt:term 2 check}a) and (\ref{eqt:term 2 check}c) converge to $0$ in probability when $N\rightarrow \infty$ using the arguments in proving that (\ref{eqt:term 1 check}a) and (\ref{eqt:term 1 check}c) converge to $0$. As a result, the quantity (\ref{eqt:consistent proof estimator decomp}d) converges to $0$ in probability when $N\rightarrow \infty$.
	
Lastly, we turn to consider the quantity (\ref{eqt:consistent proof estimator decomp}a). In fact, we have
\begin{subequations}
{\small
\begin{align}
&\mathbb{P}_{\hat{\rho}}\left\{\left|\hat{\kappa}_{R,N}^{i;F}-\kappa_{R,N}^{i;F}\right|\geq\frac{\epsilon}{8}\right\}\nonumber\\
\leq&\mathbb{P}_{\hat{\rho}}\left\{\left|\frac{1}{N}\underset{m\in \mathscr{I}^{c}}{\overset{}{\sum}}\frac{1}{R}\underset{u=1}{\overset{R}{\sum}}\left(\hat{\xi}_{m,u}^{i;F}\right)(\hat{A}_{m}^{i}-A_{m}^{i})\right|\geq\frac{\epsilon}{16}\right\}\label{eqt:Chebyshev Inequality quantity 2 term 1}\\
&\;+\mathbb{P}_{\hat{\rho}}\left\{\left|\frac{1}{N}\underset{m\in \mathscr{I}^{c}}{\overset{}{\sum}}A_{m}^{i}\frac{1}{R}\underset{u=1}{\overset{R}{\sum}}\left[(\hat{\xi}_{m,u}^{i;F}-\xi_{m,u}^{i;F})\right]\right|\geq\frac{\epsilon}{16}\right\}.\label{eqt:Chebyshev Inequality quantity 2 term 2}
\end{align}
}\noindent
\end{subequations}
We can argue that \eqref{eqt:Chebyshev Inequality quantity 2 term 1} converges to $0$ in probability as $N\rightarrow \infty$ using similar arguments when we prove that (\ref{eqt:consistency check equation}b) converges to $0$ in probability. Simultaneously, we can argue \eqref{eqt:Chebyshev Inequality quantity 2 term 2} converges to $0$ in probability as $N\rightarrow \infty$ using similar arguments when we prove that (\ref{eqt:consistency check equation}c) converges to $0$ in probability. Consequently, we have $\hat{\kappa}_{R,N}^{i;F}-\kappa_{R,N}^{i;F}$ converges to $0$ in probability.

The proof is completed.
\end{proof}

\bibliographystyle{IEEEtran}
\bibliography{ijcnn}


\end{document}